\documentclass{article}

\usepackage[reset,a4paper,vmargin=3.5truecm,hmargin=3truecm]{geometry}

\usepackage{graphicx}        % Include this line if your 

\usepackage{amsmath,amsxtra,amssymb,latexsym, amscd,amsthm,epic,bbm,bm,mathtools,leftidx}

\usepackage{array}
\usepackage{verbatim}

\usepackage{subcaption}
\usepackage{tikz}
\usepackage{hyperref}

\graphicspath{{Figures//}}

\DeclareMathOperator{\Spec}{Spec}

\DeclareMathOperator{\tridiag}{tridiag}
\newcommand{\Tridiag}[4]{\tridiag\mat{#1 & #2 \\ #3}_{#4}}

\newcommand{\cp}[2]{P^{(#1)}_{#2}} %

\newcommand{\AHRabi}[1]{H_{#1}} % A. Q. Rabi Model Hamiltonian
\newcommand{\tHRabi}[2]{H_{#1}^{(#2)}} % A. Q. Rabi Model Hamiltonian

\newcommand{\N}{\mathbb{N}} % natural numbers
\newcommand{\Z}{\mathbb{Z}} % integers
\newcommand{\R}{\mathbb{R}} % real numbers
\newcommand{\C}{\mathbb{C}} % complex numbers
\newcommand{\uH}{\mathbb{H}} % upper half plane
\newcommand{\e}{\epsilon}

\newcommand{\id}{{1\!\!1}}

\theoremstyle{plain}
\newtheorem{thm}{Theorem}[section]

\newtheorem{prop}[thm]{Proposition}

\newtheorem{cor}[thm]{Corollary}

\newtheorem{conj}[thm]{Conjecture}

\theoremstyle{definition}
\newtheorem{dfn}{Definition}[section]

\theoremstyle{remark}
\newtheorem{rem}{Remark}[section]

\newcommand{\mat}[1]{\begin{bmatrix}#1\end{bmatrix}}

\title{Spacing distribution for quantum Rabi models}

\author{Daniel Braak, Linh Thi Hoai Nguyen, Cid Reyes-Bustos and Masato Wakayama}

\begin{document}

\maketitle

\begin{abstract}
  The asymmetric quantum Rabi model (AQRM) is a fundamental model in quantum optics describing the interaction of light and matter. Besides its immediate physical interest, the AQRM possesses an intriguing mathematical structure which is far from being completely understood. In this paper, we focus on the distribution of the level spacing, the difference between consecutive eigenvalues of the AQRM in the limit of high energies, i.e. large quantum numbers. In the symmetric case, that is the quantum Rabi model (QRM), the spacing distribution for each parity (given by the $\Z_2$-symmetry) is fully clarified by an asymptotic expression derived by de Monvel and Zielinski, though some questions remain for the full spectrum spacing. However, in the general AQRM case, there is no parity decomposition for the eigenvalues. In connection with numerically exact studies for the first 40,000 eigenstates we describe the spacing distribution for the AQRM which is characterized by a new type of periodicity and symmetric behavior of the distribution with respect to the bias parameter. The results reflects the hidden symmetry of the AQRM known to appear for half-integer bias.  In addition, we observe in the AQRM the excited state quantum phase transition for large values of the bias parameter, analogous to the QRM with large qubit energy, and an internal symmetry of the level spacing distribution for fixed bias. This novel symmetry is independent from the symmetry for half-integer bias and not explained by current theoretical knowledge. 


\end{abstract}

\tableofcontents

%% 81V73

%%%%%%%%%%%%%%%%%%%%%%%%%%%
% Section 1: Introduction %
%%%%%%%%%%%%%%%%%%%%%%%%%%%

\section{Introduction}
\label{sec:Introduction}

The quantum Rabi model (QRM), along with its generalizations, are fundamental, and at the same time, one of the simplest models for quantum interaction (see \cite{bcbs2016} and references therein). It is the fully quantized model given in \cite{JC1963} of the (semiclassical) Rabi model introduced by Isidor Issac Rabi in 1936.  The QRM describes the dipolar coupling between a two-level system and a single bosonic field mode. In this way, it supplies the basis for the theoretical description of many important structures in quantum optics and condensed matter physics, such as cavity quantum electrodynamics (CQED), quantum dots, polarons, trapped ions and solid state (circuit) QED (cQED). Partly due to its numerous applications in platforms dedicated to the implementation of quantum information devices \cite{HR2008}, it has been intensely studied both theoretically  and experimentally \cite{Le2016, Y2017, YS2018}. More recently, there has been a number of studies focused on the mathematical aspects of the QRM \cite{HH2012, Sugi2016, KRW2017, RW2019, BZ2021,RW2021, RW2023CMP}.

The subject of this paper is the eigenvalue spacing distribution of a simple but significant generalization of the QRM known as the asymmetric quantum Rabi model (AQRM), obtained by adding a static bias to the two-level system which flips the pseudo-spin spontaneously.
The Hamiltonian of the AQRM is given by
\begin{equation}
  \label{eq:AQRMHamil}
  \AHRabi{\e} := \omega a^{\dagger}a + \Delta \sigma_z + g (a + a^{\dagger}) \sigma_x + \epsilon \sigma_x,
\end{equation}
where $a^\dag$ and $a$ are the creation and annihilation operators of the bosonic mode, i.e. $[a,\,a^\dag]=1$ and
\[
\sigma_x = \begin{bmatrix}
 0 & 1  \\
 1 & 0
\end{bmatrix}, \qquad
\sigma_z= \begin{bmatrix}
 1 & 0  \\
 0 & -1
\end{bmatrix}
\]
are the Pauli matrices, $2\Delta$ is the energy difference between the two levels, $g$ denotes the coupling strength between the two-level system and the bosonic mode with frequency $\omega>0$, and the parameter \(\e \in \R\) measures the strength of the bias. The QRM, i.e. symmetric QRM, corresponds to the case $\e=0$. The term $\e \sigma_x$ in the Hamiltonian $\AHRabi{\e}$ appears naturally in implementations with flux qubits \cite{Ni2010} (see \cite{Y2017, YS2018} for a recent important experiment using this implementation to reach the deep strong coupling regime of the QRM and AQRM). The AQRM has been studied from the viewpoint of topology in \cite{LB2021b} and with respect to entanglement in \cite{SCE2022}.

This flip term given by the bias parameter $\e$ in the AQRM breaks the natural $\Z_2 (= \Z/2\Z)$  or ``parity'' symmetry which is characteristic for the QRM and renders it integrable \cite{B2011PRL}. The absence of this symmetry makes the spectrum of the AQRM in general multiplicity free. However, when the parameter $\e$ is a half-integer multiple of $\omega$ (set to 1 in the following), the spectrum shows degeneracies very similar to the QRM with $\e=0$. These degeneracies are actually related to a ``hidden'' symmetry of the AQRM which cannot be read off from the Hamiltonian directly. 
The existence of the spectral degeneration, whence the hidden symmetry, has a significant impact on the spectral structure of the system. For instance, the presence of the hidden symmetry renders the AQRM with half-integer $\e$ integrable in the quantum sense like the QRM. Indeed, the notion of quantum integrability proposed in \cite{B2011PRL} can be connected to the well-known concept of Yang-Baxter integrability \cite{E2019} by using the concept of ``analytic commutant'' introduced in \cite{RBBW2021}. In Section \ref{sec:preliminaries} we give an overview of the hidden symmetry and degeneracy of the AQRM.

To analyze qualitative aspects of the spectrum and its asymptotic structure, important information may be obtained from the distribution of the difference between adjacent eigenvalues (or level spacing distribution) and other statistical quantities. For instance, the distribution of eigenvalues of symmetric or Hermitian random matrices (GOE and GUE ensembles, respectively \cite{M2004}) resembles the spacing distribution of the zeros of the Riemann zeta function (see \cite{KS1999} and the references therein). %\cite{S1997,S1997b}.
  Actually, from  the mathematical viewpoint, some of the major motivations for research on the QRM are group theoretical symmetries and arithmetic properties of its Hamilton operator. This is natural since the QRM and its relatives feature both noncommutativity and discreteness, making them a reasonable research setting for non-abelian groups (and Lie algebras) and number theory.

  When considering such a Hamilton operator, the eigenvalue problem is generally quite complicated, hence finding the eigenvalues explicitly is considerably difficult. Moreover, even if the exact spectrum in terms of the spectral determinant can be found, only the simplest cases, such as the harmonic oscillator, allow for expressions in closed form. In all other cases, the study of the statistical distribution of the eigenvalues may yield valuable insight into qualitative aspects of the spectrum. In this regard, the spectral zeta function and the spacing distribution of consecutive eigenvalues are particularly useful. (Let us note that the spacing distribution of consecutive zeros of several  zeta functions are very important subjects in number theory \cite{KS1999}.) Concerning the former, of particular importance are the study of the analytic properties (see, e.g. \cite{IW2005a, Sugi2016, RW2021}), the derivation of the Weyl law and the special values (see \cite{KW2023} and the reference therein, and e.g. \cite{KZ2001} for further mathematical implications). For the latter, there are many studies within random matrix theory which have uncovered a relation of the level distribution to the integrability of the system.

  For instance, there is a well-known phenomenological observation (see, e.g. \cite{S2003}) about the spectrum of the hyperbolic Laplacian $\Delta_{\uH}$ (commuting with the action of $SL_2(\R)$) on the hyperbolic space $SL_2(\Z)\backslash \uH$, $\uH$ being the upper half-plane $SL_2(\R)/SO(2)$ equipped with the Poincar\'e metric (as a Riemann symmetric space).
  We denote by $\lambda_j$ the (discrete) eigenvalues of $\Delta_{\uH}$ given in increasing order. Since the volume (area) of the space $SL_2(\Z)\backslash \uH$ is $\frac1{12}$ in the standard normalization, by setting $\tilde{\lambda}_j=\frac1{12}\lambda_j$ the Weyl law for the spectrum shows that the mean spacing between the numbers $\tilde{\lambda}_j$ is $1$. Then, numerical experiments suggest that the consecutive spacing (or the local spacing statistics) follows a Poisson distribution, which was initially unexpected. More precisely \cite{St1992}, for $\alpha \geq 0$ the consecutive spacing apparently satisfies the law
  \begin{equation}
    \label{eq:Laplace}
    \lim_{N\to\infty} \frac{\#\{j\leq N\, |\, 0\leq \tilde{\lambda}_{j+1} -\tilde{\lambda}_j \leq \alpha \}}N = \int_0^\alpha e^{-t}dt.
  \end{equation}
It has been verified numerically that the consecutive spacing distribution of the eigenvalues
  of  $\Delta_{\uH}$ for arithmetic subgroups $\Gamma$ of $SL_2(\R)$ as the case of $SL_2(\Z)$ appear to follow the Poisson distribution while that of non-arithmetic groups
  follow the GOE distribution \cite{S1991}. Note that the study is directly related to the consecutive zeros of the Selberg zeta
  function (see, e.g. \cite{KS1999}) which counts the number of prime geodesics in $\Gamma\backslash \uH$ via the trace formula for an arithmetic subgroup $\Gamma(\subset SL_2(\R))$ or its determinant expression (see  \cite{S1987} and also \cite{KuroW2004}). 
  
We note that in the previous examples the hyperbolic Laplacian $\Delta_{\uH}$ is essentially given by a Gaussian hypergeometric differential operator \cite{Sugiura1975, L1985} (given by the Casimir element, i.e., the 2nd order central element of the universal enveloping algebra of the Lie algebra $\mathfrak{sl}_2(\R)$), in other words, the geometry is described by a Gaussian hypergeometric ODE with three regular singularities. However, the non-commutative harmonic oscillator (NCHO \cite{PW2001}) has a Heun ODE \cite{SL2000} picture in the sense that the eigenvalue problem is equivalent to the existence of a holomorphic solution in some domain (depending on the coupling parameters in the Hamiltonian) for a certain Heun ODE with four regular singular points \cite{O2001, W2015IMRN}. Also, it is well-known that the QRM has a confluent Heun ODE picture with two regular singularities and an irregular one of degree two. Based on the corresponding ODE pictures, the NCHO may be regarded as a covering model of the QRM in the sense of \cite{RW2023CMP}. In this study, using the result of \cite{BZ2021} on the oscillatory behavior of the large eigenvalues of the QRM for each parity (determined by the $\Z_2$-symmetry of the Hamiltonian) we observe that the spacing
distribution of consecutive eigenvalues for the QRM is neither a Poisson distribution nor does it follow the GOE or GUE distributions which would have been an indicator for the chaotic (non-integrable) nature of the system. Indeed, it can be argued that the QRM is quantum integrable solely due to its  $\Z_2$-symmetry \cite{B2011PRL}.  

An early argument \cite{kus1985} invoked against the integrability of these models stems from the statistical analysis of the level spacing distribution, because it is not of Poissonian type as required by the Berry-Tabor criterion \cite{BT1977}. However, as this criterion does only apply to systems with \textit{at least two} continuous degrees of freedom, it says nothing about the QRM and AQRM which possess only one of these (the bosonic mode). Thus, the common tool to discern between integrable and chaotic quantum dynamics is not applicable in the present circumstances. A detailed study of the level spacing distribution in the (A)QRM to see the impact of the symmetry is the subject of the present paper.

The paper is structured as follows. In Section \ref{sec:preliminaries} we present the necessary background on the AQRM, our numerical method is discussed in Section \ref{sec2:EigenvalueComputation} and Section \ref{sec:spacing} contains the main results and findings. Namely,  based on numerical experiments we give the expression of the density of the energy spacing for the QRM and prove\footnote{In the first draft of this paper (arXiv:2310.09811 [math-ph] 26 October 2023) this was a conjecture. The authors came across the paper \cite{cfz2023} after uploading the manuscript to the arXiv.}  the distribution for the AQRM using recent results \cite{cfz2023} on the asymptotic distribution for the AQRM. From this, we see that the spectrum itself and the spacing distribution have a symmetric structure and a peculiar periodicity with respect to the parameter $\e$, indicating the presence of the hidden symmetry for half-integer values. Furthermore, the level spacing distribution for fixed $\e$ shows a kind of reflection symmetry of yet unknown origin. Section \ref{sec:furtherobservations} discusses the case of large $\e$ and the signature of an excited state quantum phase transition in the AQRM, along with a numerical study of the order of the error term of the approximation in the asymptotic estimates for the high-energy eigenvalues in the QRM given in \cite{BZ2021}.

%%%%%%%%%%%%%%%%%%%%%%%%%%%%%%%%%%%%%%%%%%%%%
% Section 2: Asymmetric quantum Rabi models %
%%%%%%%%%%%%%%%%%%%%%%%%%%%%%%%%%%%%%%%%%%%%%

\section{Asymmetric quantum Rabi model}
\label{sec:preliminaries}

In this section, we  recall some basic facts on the AQRM and
its spectrum.
As mentioned in the introduction, the Hamiltonian of the AQRM defined in \eqref{eq:AQRMHamil} acts on an infinite-dimensional  Hilbert
space $\mathcal{H}$ for parameters $\omega, g,\Delta>0$ and $\e \in \R$. Throughout the paper, we set $\omega=1$ without loss of generality. 
Note that
\[
  \Spec(\AHRabi{\e}) = \Spec(\AHRabi{-\e}),
\]
and we assume $\e \geq 0$ with no loss of generality. On occasion, we assume the
realization $\mathcal{H} = L^2(\R)\otimes \C^2$ to simplify the arguments.
We denote the spectrum of the AQRM by $\lambda_n = \lambda_n(g,\Delta,\e) \in \R, \, n=1,2,3,\ldots$ with
\[
  \lambda_1 < \lambda_2\leq \ldots \leq \lambda_n \leq \ldots \nearrow +\infty.
\]
In general, we omit the dependence on parameters if it is clear from the context.
Note that the ground state corresponding to the first eigenvalue $\lambda_1$ is simple as for the QRM \cite{HH2012}.

In Figure \ref{fig:Eigencurves1}, we illustrate the spectral curves for the AQRM for different values of $\epsilon \geq 0$.
For better visibility, we show the spectral curves (or level lines) of the renormalized Hamiltonian $\AHRabi{\e}+g^2$, where the asymptotic spectral curves are horizontal. It is immediately clear from the figures that for certain values of $\e$ there are crossings in the spectral curves. We discuss the details of this spectral degeneracy (energy level crossing) in Section \ref{sec:degeneracysymmetry} below by summarizing the results obtained in \cite{KRW2017, RBBW2021, RW2022}.

\begin{figure}[h!]
  \begin{center}
    \begin{subfigure}[b]{0.32\textwidth}
    \centering
    \includegraphics[height=3.5cm]{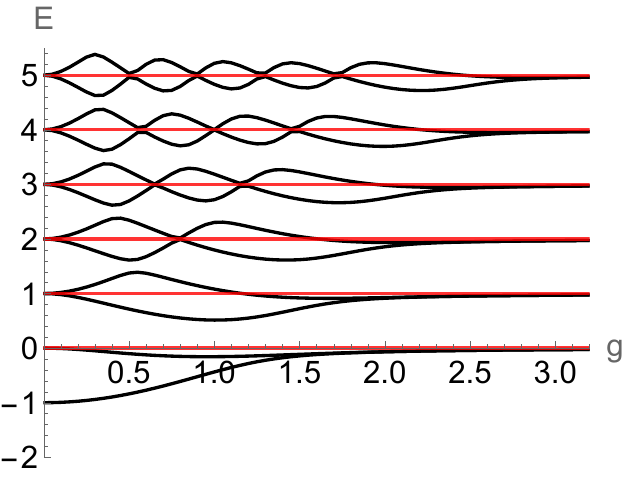}
    \caption{\(\e = 0\)}
  \end{subfigure}
  ~
    \begin{subfigure}[b]{0.32\textwidth}
    \centering
    \includegraphics[height=3.5cm]{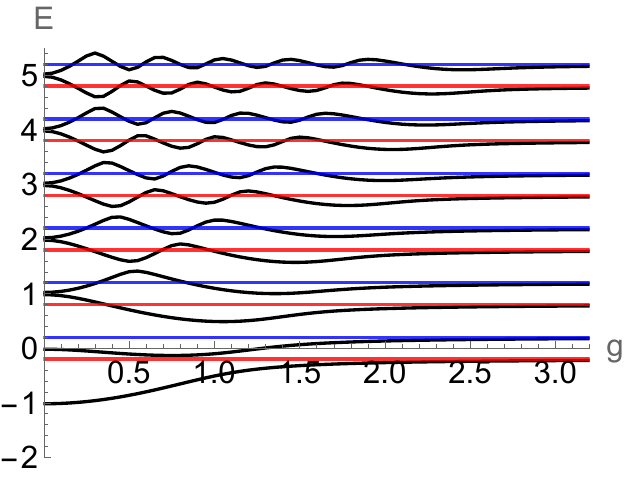}
    \caption{\(\e = 0.2\)}
  \end{subfigure}
    ~
    \begin{subfigure}[b]{0.32\textwidth}
    \centering
    \includegraphics[height=3.5cm]{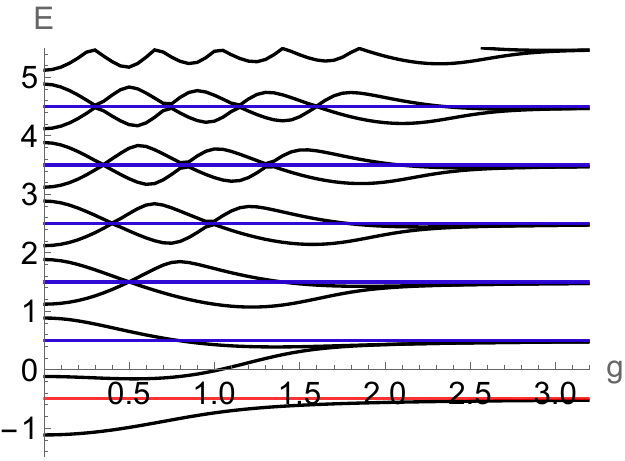}
    \caption{\(\e = 0.5\)}
  \end{subfigure}
 \end{center}
  \caption{Renormalized spectral curves of the Hamiltonian $\AHRabi{\e}+g^2$ for $\e = 0,0.2,0.5$ with
  exceptional lines (baselines) $y=n - \e$ (red) and $y=n + \e$ (blue) for $n=0,1,\cdots,5$.}
  \label{fig:Eigencurves1}
\end{figure}

Due to the difficulty to obtain rigorous results on the energy spacing for the AQRM with finite $g$, it is convenient to recall 
the weak limit $g\to \infty$ where the behavior of the spectrum is known explicitly. 
In \cite{RW2023} it was shown using spectral zeta function methods that the spectrum of the AQRM in the limit
$g\to \infty$ is given by
\begin{equation}
    \label{eq:limitspec}
    n \pm \e,
\end{equation}
for $n\geq 0$. As shown in Figure \ref{fig:Eigencurves1}, the spectral curves resemble the behavior of the limit $g \to\infty$ even for relatively small $g$ (small from the numerical viewpoint; any physical realization of these coupling values corresponds to the ``deep strong" coupling regime of the AQRM \cite{Y2017, YS2018}). 
This result was also recently proved by Fumio Hiroshima using resolvent methods \cite{H2023}, and has been considered in the physics literature as the ``natural" behavior for large coupling, albeit a formal proof has never been explicitly given (see e.g. \cite{S1985}). 

\subsection{Spectral degeneracy and symmetry}
\label{sec:degeneracysymmetry}

  A well-known property of the spectrum of the QRM is that the level lines in the spectral graph exhibit crossings between eigenvalues belonging to different parities \cite{K1985JMP} (illustrated in Figure \ref{fig:Eigencurves1}(a)).
  In fact,  in the symmetric case $\e=0$ there is a $\Z_2$-symmetry (parity) given by the involutive operator $J_0= \mathcal{P} \sigma_z$, with $\mathcal{P}= \exp(2 \pi i a^\dag a)$, such that
\[
  [\AHRabi{0} ,J_0]= 0.
\]
This leads to a decomposition of the Hilbert space $\mathcal{H}$ 
\[
  \mathcal{H} \simeq \mathcal{H}_+ \oplus \mathcal{H}_-
\]
into invariant subspaces $\mathcal{H}_{\pm} \simeq L^2(\R)$. The absence of level repulsion is obviously due to the invariance of $\mathcal{H}_\pm$ with respect to $\AHRabi{0}$.

The Hamiltonian $H_{\pm}$ is the projection of $\AHRabi{0}$ onto the subspace $\mathcal{H}_\pm$ and given by
\[
  H_{\pm} =  a^{\dag} a + g ( a + a^{\dag}) \pm \Delta \hat{T},
\]
where $\hat{T}$ defined by $(\hat{T}\psi)(x):= \psi(-x)$  for $\psi(x) \in L^2(\R)$ is the reflection operator acting on \( L^2(\R)\).
According to this decomposition, we call the eigenstates of $H_+$ (resp. $H_-$) to have positive (resp. negative)
parity.

Note that in Figure \ref{fig:Eigencurves1}(b) with  $\e = 0.2$ there are no spectral crossings and in Figure \ref{fig:Eigencurves1}(c) the crossings appear again for $\e = 0.5$. In general, it is known that when $\e \ne 0$ there are no longer crossings in the spectral curves except for the case $\e \in \frac12 \Z$. In the half-integer case $\e = \frac{\ell}2$ ($\ell \in \mathbb Z$), the symmetry operator was finally found \cite{MBB2020} for small epsilon and investigated from a mathematical viewpoint in \cite{RBBW2021, RW2022}. Concretely, it was shown in \cite{RBBW2021} that there is a unique (up to a scalar multiple) self-adjoint operator $J_{\ell}$ such that
\[
  [\AHRabi{\frac{\ell}2},J_{\ell}]=0
\]
and $J_\ell$ generates together with $\AHRabi{\frac{\ell}2}$  the analytic commutant of $\AHRabi{\frac{\ell}2}$. The question arises whether the presence of this operator which is ``analytically'' independent from $\AHRabi{\frac{\ell}2}$ \cite{RBBW2021}, indicates a symmetry of $\AHRabi{\frac{\ell}2}$, similar to $\AHRabi{0}$.

A key property of the operator $J_\ell$ is that it satisfies
\begin{equation} \label{eq:Jquad2}
  J_{\ell}^2 = p_{\ell}(\AHRabi{\frac{\ell}2};g,\Delta)
\end{equation}
for a polynomial $p_{\ell} \in \R[x,g,\Delta]$ of degree $\ell$. However, the operator $J_{\ell}$ itself is clearly not a polynomial function of $\AHRabi{\frac{\ell}2}$ over $\C$. Thus,  the ``analytic'' commutant  algebra of  $\AHRabi{\frac{\ell}2}$ is given by the polynomial ring  $\C[\AHRabi{\frac{\ell}2},J_\ell]$.  Note that $p_{\ell}(x;g,\Delta)\geq0$ for any $x\in \R$ and the equation $y^2=p_{\ell}(x;g,\Delta)$ defines in general a hyperelliptic curve on the $(x,y)$-plane \cite{RW2022}.

But $J_\ell$ is not involutive, so it does not fulfill the $\Z_2$ group relation $J_\ell^2=\id$. On the other hand, if an eigenvector $\psi_\lambda$ of $\AHRabi{\frac{\ell}2}$ with eigenvalue $\lambda$ is not an eigenvector of $J_\ell$, the state $J_\ell\psi_\lambda=\psi'_\lambda$ is an eigenvector of $\AHRabi{\frac{\ell}2}$ with the same eigenvalue $\lambda$. But $J_\ell\psi'_\lambda =\alpha\psi_\lambda$ with $\alpha>0$ due to relation \eqref{eq:Jquad2}, so the degenerate subspace $\mathcal{H}_\lambda$ is just two-dimensional. After diagonalization, $J_\ell$ acts in $\mathcal{H}_\lambda$ as a multiple of $\sigma_z$ which satisfies $\sigma_z^2=\id$. It follows that a suitably normalized $J_\ell$ provides a representation of $\Z_2$ in $\mathcal{H}_\lambda$.
For a non-degenerate eigenvalue $\lambda \in \Spec(\AHRabi{\frac{\ell}2})$ we may define the parity of $\psi_\lambda$ as the sign of the associated eigenvalue $\mu(\lambda)$ of $J_\ell$. 

A problem arises if for some non-degenerate eigenvalue $\lambda_0$ the corresponding eigenvalue of $J_\ell$ vanishes, $\mu(\lambda_0)=0$. At present, it is unknown whether this case may occur or not. If so, the parity of a  spectral line could change discontinuously if a model parameter is varied and would be undefined for $\psi_{\lambda_0}$.
However, it is not difficult to show that for large enough $\lambda \in \Spec(\AHRabi{\frac{\ell}2})$, the parity is well-defined. Actually we have the following. 

\begin{prop}
  \label{prop:parity}
  For $\e = \frac{\ell}2$ with integer $\ell$, except for at most a finite number of eigenvalues we can assign parities to the eigenvalues such that spectral crossings occur only between eigenvalues of different parities.
\end{prop}

\begin{proof}
  
  In general, this problem is equivalent to showing that the kernel of the operator $J_{\ell}$ is trivial, which is
  equivalent to showing that all the eigenvalues $\lambda \in \Spec(\AHRabi{\frac{\ell}2})$ satisfy
  \[
    p_\ell(\lambda;g,\Delta) \neq 0.
  \]
  
  Since $p_\ell(x;g,\Delta)$ is a polynomial of degree $\ell$ with respect to $x$, the set of roots is bounded. Therefore, there exists a constant $c>0$ such that \(p_\ell(\lambda;g,\Delta) \neq 0\) for all $\lambda \in \Spec(\AHRabi{\e})$ satisfying $|\lambda| > c$.
\end{proof}

In practice, determining the parity of an eigenvalue is non-trivial since it depends on explicit information on the polynomial
$p_\ell(\lambda;g,\Delta)$ (see the discussion in \cite{RW2022}). However, under certain assumptions, it is possible to assign the parity at $g=0$ and extend it along the spectral lines (see \cite{RW2022}). We may then verify that this method coincides with the (parity invariant) subspace decomposition for the case $\e=0$. 

\subsubsection{Constraint polynomials}
\label{sec:constraint}

As mentioned in the introduction, numerical investigations suggested that for half-integer bias, degeneracies appear in the quasi-exact (Juddian) part of the spectrum of the AQRM \cite{LB2015JPA, BLZ2015} and a rigorous proof was given in \cite{KRW2017} using 
constraint polynomials. In this section, we recall some of the properties of constraint polynomials for the AQRM and we refer to the aforementioned papers for the detailed discussion. 

The constraint polynomials $\cp{N,\frac{\ell}{2}}{N}(u,v)$  are such that a zero, with respect to the variables $u=(2g)^2$ and $v=\Delta^2$, of the equation
\[
\cp{N,\frac{\ell}{2}}{N}(u,v) = 0
\]
gives a Juddian (quasi-exact) solution of the AQRM with the energy (eigenvalue) $x=N-g^2+\frac{\ell}2$. 
In \cite{KRW2017} it was shown that the constraint polynomials satisfy the equation
\begin{align}
  \label{eq:div}
  \cp{N+\ell,-\ell/2}{N+\ell}(u,v) =  A^\ell_N(u,v) \cp{N,\ell/2}{N}(u,v),
\end{align}
where \( A^\ell_N(u,v)\) is a polynomial in $(u,v)$ satisfying \( A^\ell_N(u,v) >0\) for $u, v > 0$. Then, by the elementary relation 
\[
    (N+\ell)-g^2-\frac{\ell}2=N-g^2+\frac{\ell}2,
\]
we find that the degeneracy occurs if and only if $u=(2g)^2$ and $v=\Delta^2$ satisfy $\cp{N,\frac{\ell}{2}}{N}(u,v)= 0$.  

In addition, it is known \cite{KRW2017} that the polynomial $A^\ell_N(u,v)$ has the determinant expression
\begin{equation}
    \label{eq:Adetexp}
    A^\ell_N(u,v) = \frac{(N+\ell)!}{N!}\det\Tridiag{u+\frac{v}{N+i}-\ell+2i-1}{1}{-i(\ell-i)}{1\le i\le \ell},
\end{equation}
where we used the notation
\begin{equation*}
  \Tridiag{a_i}{b_i}{c_i}{1\le i\le n}
  :=\begin{bmatrix}
    a_1 & b_1 & 0 &  \cdots & 0    \\
    c_1 & a_2 & b_2 &  \cdots  & 0\\
    \vdots & \ddots   & \ddots &  \ddots & \vdots   \\
    0 &  \cdots &  0  & a_{n-1} & b_{n-1} \\
    0 & \cdots  & 0  & c_{n-1} & a_n
  \end{bmatrix}
\end{equation*}
for a tri-diagonal matrix.  We note here that the constraint polynomials and $\cp{N,\ell}{N}(u,v)$ appear to be closely related \cite{RW2022} to, and in fact have been used to approximate, the full spectrum of the AQRM as shown in the generalized adiabatic approximation \cite{LB2021b}.

In the next section, we provide a brief overview of the truncated Hamiltonian method, the standard method to compute the spectrum of the AQRM and other models in quantum optics. 

%%%%%%%%%%%%%%%%%%%%%%%%%%%%%%%%%%%%%%%%%%%%%%%%%%%%%%%%
% Section 3: Truncated Hamiltonian method for the AQRM %
%%%%%%%%%%%%%%%%%%%%%%%%%%%%%%%%%%%%%%%%%%%%%%%%%%%%%%%%

\section{Truncated Hamiltonian method for the AQRM}
\label{sec2:EigenvalueComputation}

In general, it is difficult to determine the energy levels by solving the Schr\"odinger equation 
\[
    H \phi_n(x)=E_n \phi_n(x),
\]
where $H$ is the Hamiltonian operator of the system, $\phi_n(x)$ are wave functions that represent the
state of the system in the position basis, and $E_n$ are the corresponding energy levels (eigenvalues).

The truncation of the Hamiltonian is a numerical method to study eigenvalue problems in infinite-dimensional vector spaces, known as the Rayleigh-Ritz method in quantum mechanics, or simply as ``exact diagonalization'' \cite{WH2008}. 
 
Given a self-adjoint operator $H$ in $\mathcal{H}$, the method is based on finding finite-rank Hamiltonians $H^{(N)}$ by projecting (``truncating'') $H$ onto a finite-dimensional subspace of $\mathcal{H}$. The  $H^{(N)}$ should ``approximate'' $H$ so that the spectrum of $H^{(N)}$, computed using standard numerical tools, approximates the spectrum of $H$. In general, an operator $H$  cannot be approximated by finite-rank operators unless it is a compact operator \cite{RS1981}. The QRM and the AQRM are not compact operators but it has been shown that the truncation method is appropriate to compute their eigenvalues \cite{Braak2013}.

Before considering the general case, let us first treat the computation of the eigenvalues for the QRM, that
is the Hamiltonian $\AHRabi{0}$. In this case, since we have the invariant decomposition of $\mathcal{H}$, we can consider
truncated state spaces for each of the parities.

The projection operators $\hat P^{(N)}_{\pm} : \mathcal{H}_{\pm} \to \mathcal H^{(N)}_{\pm} \simeq \C^{N+1}$ are used to project
$\AHRabi{0}$ onto $\mathcal{H}_{\pm}$. The matrix expansion with respect to the eigenbasis of $\mathcal{H}_\pm$ given by 
eigenfunctions of $a^\dag a$ is a tridiagonal matrix
\begin{equation}\label{QRM_TruncatedHamiltonians}
  H^{(N)}_{\pm} = 
  \begin{pmatrix}
    \pm \Delta & g & 0 & \ldots & 0  & 0\\
    g & 1 \pm (-\Delta) &  \sqrt{2}g  & \ldots &  0 & 0 \\
    0&  \sqrt{2}g  &2 \pm\Delta&\ldots & 0  & 0\\
    \vdots &\vdots &\vdots &\ddots &\vdots  & \vdots \\
    0 &0  &0 & 0 &  (N-1)  \pm (-1)^{N-1} \Delta&  \sqrt{N} g\\
    0& 0 &0 & 0 &  \sqrt{N} g & N  \pm (-1)^{N}\Delta
  \end{pmatrix}.
\end{equation}

The spectrum of the QRM can be approximated from the eigenvalues of this matrix provided that the Hamiltonian is truncated at a sufficiently large dimension.

For the general case $\e \neq 0$, since there is no known invariant decomposition, we consider projections of $\mathcal{H}$ into
subspaces $\mathcal{H}^{(N)} \simeq \C^{N+1} \otimes \C^2$. We can rewrite the Hamiltonian of the AQRM in \eqref{eq:AQRMHamil} as
\begin{align*}
    \AHRabi{\e}
    &=\begin{pmatrix}
    a^\dagger a+\Delta &g(a^\dagger+a)+\epsilon \\
    g(a^\dagger+a)+\epsilon & a^\dagger a-\Delta 
\end{pmatrix}.
\end{align*}

The Hamiltonian of the AQRM can be approximated by the matrix of size $2(N+1)$, given in block form as
\[
  \tHRabi{\e}{N} =\left[
    \begin{array}{c|c}
      M_1 & M_2 \\ \hline
      M_2 & M_3
    \end{array}\right],
\]
where $M_1 = D_{N+1} + \Delta I_{N+1}$, $M_3 = D_{N+1} - \Delta I_{N+1}$ with $D_N = {\rm diag}(0,1,\cdots,N-1)$, $I_N$ is the
identity matrix, and
\[
  M_2 =
  \begin{pmatrix}
    \epsilon&g\sqrt 1&0&\ldots&0&0\\
    g\sqrt 1&\epsilon&g\sqrt 2&\ldots&0&0\\
    0&g\sqrt 2&\epsilon&\ldots&0&0\\
    \vdots&\vdots&\vdots&\ddots&\vdots&\vdots\\
    0 & 0 & 0 & \ldots &\epsilon& g\sqrt{N}\\
    0&0&0&\ldots&g\sqrt N&\epsilon\\
  \end{pmatrix}.
\]

The accuracy of the eigenvalues for a given level $N$ improves as we increase the dimension of the matrix. For the computations in this paper, we are able 
to approximate the  first $40,000$ eigenvalues in the spectrum of the AQRM, the number being limited by the memory of
the computer (2.5 GHz Dual-Core Intel Core i7 processor, 16 GB 2133 MHz LPDDR3 memory). 

%%%%%%%%%%%%%%%%%%%%%%%%%%%%%%%%%%%%%%%%%%
% Section 4: Energy spacing distribution %
%%%%%%%%%%%%%%%%%%%%%%%%%%%%%%%%%%%%%%%%%%

\section{Energy spacing distribution}
\label{sec:spacing}

In this section, we describe the distribution of the energy spacing of the full spectrum of the (A)QRM using
known asymptotic estimates for the QRM as a basis.

First, for system parameters $(g,\Delta,\e)$, let us define the set of energy spacing $S(g,\Delta,\epsilon)$ as 
\[
  S(g,\Delta,\e) := \{ \lambda_{n+1} - \lambda_n \, | \, \lambda_n,\lambda_{n+1} \in \Spec(\AHRabi{\e}) \}_{n\geq1}.
\]
Clearly, $S(g,\Delta,\e) \subset [0,\infty)$ and $0 \in S(g,\Delta,\epsilon)$ if and only if the spectrum is degenerate. 
In particular, note that if $\e \notin \frac12\Z$, we have $S(g,\Delta,\e) \subset (0,\infty)$. 

To study the distribution of the set $S(g,\Delta,\e)$ we first consider the symmetric case ($\e=0$),
that is, the QRM. In this case,  the spacing behavior for a single parity follows from the asymptotics of the eigenvalues,  obtained by de Monvel and Zielinski in \cite{BZ2021}.

Let us denote by $\lambda_n(H_\pm)$ the $n$-th eigenvalue of the Hamiltonian $H_\pm$ 
\[
  \lambda_1(H_\pm)< \lambda_2(H_\pm)< \cdots.
\]
Note that degeneration occurs only between the spectrum of $H_+$ and $H_-$. We reproduce here the
main result of \cite{BZ2021} in a slightly different formulation (see also \cite{BZ2017}).
In Section \ref{sec:curvefitting}, using numerical computations, we make some remarks on the error term of this approximation.

\begin{thm}[Boutet de Monvel and Zielinski, 2021]
  \label{thm:asymptotic}
  For any $\delta>0$, we have
\begin{equation}
  \label{eq:asymptotics}
  \lambda_n(H_{\pm})= n- g^2\mp (-1)^n  \Delta \frac{\cos(4 g\sqrt{n} -\frac{\pi}{4})}{\sqrt{2\pi g}} n^{-1/4} + O(n^{- 1/2+\delta})
\end{equation}
as $n \to \infty$.
\end{thm}

The theorem above immediately implies that
\begin{equation}
  \label{eq:spacing single parity}
  \lambda_{n+1}(H_\pm) - \lambda_n(H_\pm) \to  1\quad  \text{as } \quad n \to \infty
\end{equation}
and thus the energy spacing in each parity is bounded. It follows that the set $S(g,\Delta,0)$ is bounded above for
any $g,\Delta>0$.

We now consider a measure to describe the distribution of the energy spacing for each of the parities.
For $N\in \Z_{\geq 1}$ and $\alpha\geq 0$, define 
\begin{align*}
  M^{\pm}_N(\alpha)= M^{\pm}_N(\alpha;g,\Delta)&:= \# \{ n \leq N: \lambda_{n+1}(H_\pm)-\lambda_n(H_\pm)<\alpha \}, \\
  \alpha^{\pm}_0(g,\Delta) &:=\sup\limits_{n} \{ \lambda_{n+1}(H_\pm)-\lambda_n(H_\pm) \}.
\end{align*}

\begin{prop}
  \label{prop:measureParity}
  There is a measure $\mu_\pm$ in $[0,\infty)$ such that
  \[
    \lim_{N\to \infty}\frac{M^{\pm}_N(\alpha)}{M^{\pm}_N(\alpha_0^\pm)}= \int_0^\alpha \mu_\pm(x) dx.
  \]
  The measure $\mu_\pm(x)$ is given by the Dirac distribution centered in $1$, that is, 
  \[
    \mu_\pm(x) = \delta(x-1).
  \]
Note that each measure $\mu_\pm$ does not depend on the parameters $g,\Delta$ nor the parity. 
\end{prop}

\begin{proof}

  The asymptotic estimate \eqref{eq:asymptotics} suggests that for any $\eta>0$ there is $N \in \Z_{\geq0}$ such that for
  any $n>N$ we have 
  \[
    \lambda_{n+1}(H_\pm) - \lambda_n(H_\pm) \in (1 -  \eta, \, 1 + \eta). 
  \]

Hence, letting $N\to\infty$, we observe that the density distribution is a step-function of the form
  \[
    \lim_{N\to \infty}\frac{M_N^\pm(\alpha)}{M^\pm_N(\alpha_0^\pm)}= 
    \begin{cases}
      0   & \text{ for } \alpha < 1 \\
      c   & \text{ for } \alpha = 1 \\
      1   & \text{ for } \alpha > 1
    \end{cases},
  \]
  for some constant $0 \leq c \leq1$, and the result follows directly. We note that using the oscillating term in the asymptotic estimate \eqref{eq:asymptotics} it is not difficult to verify that $c=\frac12$. The distribution is neither right nor left continuous at $\alpha=1$.
\end{proof}

If we observe the full spectrum not restricted to each parity, consecutive energies may belong to different parities, that is, we may have an energy spacing of the form
\[
  \lambda_{n+1}(H_\pm) - \lambda_n(H_\mp).
\]
From Theorem \ref{thm:asymptotic}, we see that for large $n$ the value may be close to zero or to $1$, but it is difficult to determine in general due to the presence of the oscillatory and error terms. In the remainder of this section, using numerical experiments, we explore the properties of the energy distribution for the QRM and the AQRM.

Let us finish this subsection by considering the supremum  $\alpha^{\pm}_0(g,\Delta)$ of the energy spacing defined above. 
As already mentioned above, $\alpha_0^{\pm}$ is finite and the $G$-conjecture (see \cite{B2011PRL} or Section
\ref{sec:refinement} below) suggests $\alpha^{\pm}_0 \leq 3$.

In contrast with the distribution of the energy spacing, the $\alpha^{\pm}_0$ depends on the parameters $g,\Delta >0$ in a
non-trivial way. To get an idea of this dependence, we consider numerical
estimates of the parameter $\alpha^{\pm}_0$. In Figure~\ref{fig:alpha0_vs_g_Delta} we illustrate the value of $\alpha_0^+$ with
respect to the parameters $g,\Delta$. Actually, the numerical result suggests that $\alpha_0^{\pm} \leq 2$ and moreover that $\alpha_0^{\pm}$ is a
strictly decreasing (resp. increasing) function of $g$ (resp. $\Delta$).

\begin{figure}[h!]
\centering
\includegraphics[width=8cm,height=5cm]{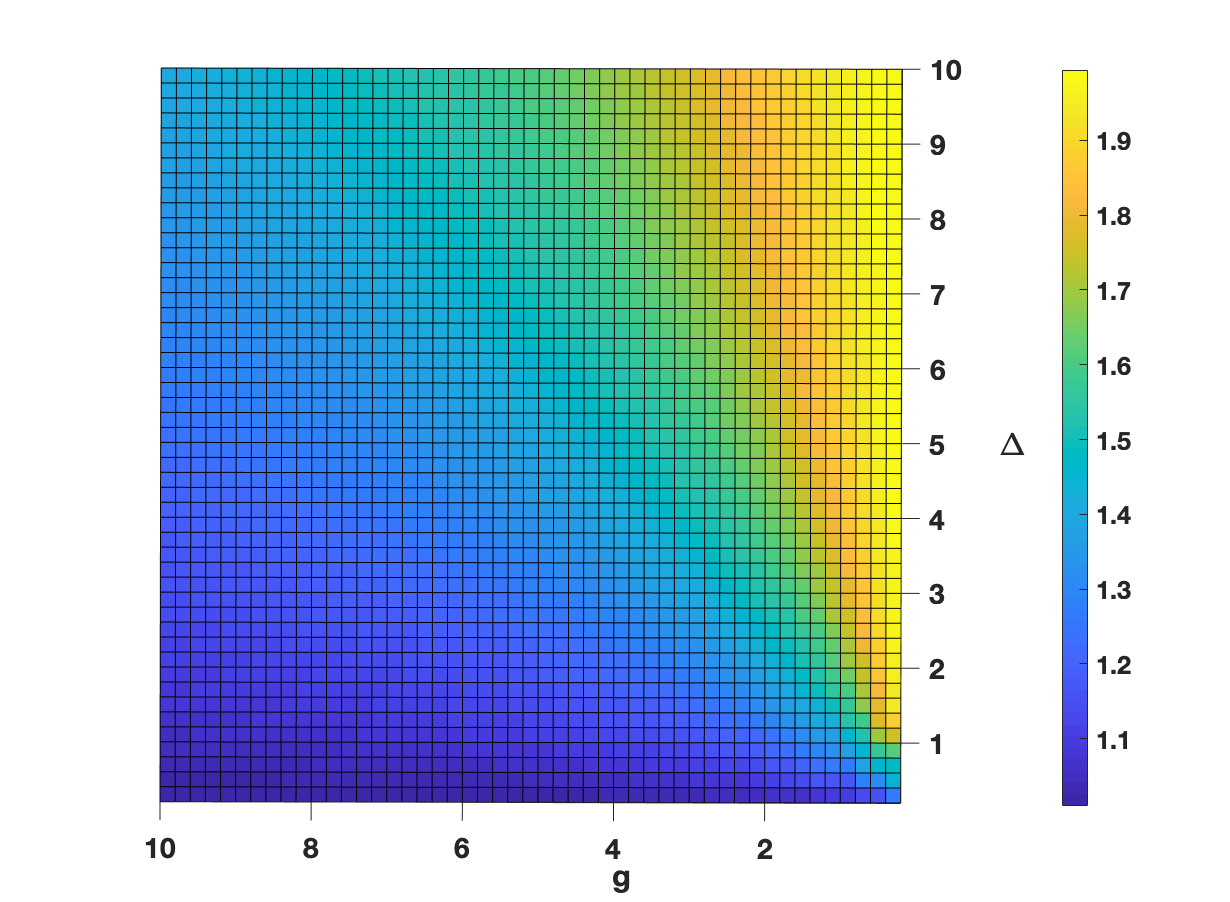}
\caption{Dependence of $\alpha_0^+$ on $g,\Delta \in (0,10]$.}
\label{fig:alpha0_vs_g_Delta}
\end{figure}

\subsection{Spacing distribution of the full spectrum}
\label{sec:fullspec}

The analysis of the energy spacing distribution for the AQRM is considerably more complicated since,
as discussed in Section \ref{sec2:EigenvalueComputation}, there is no known Hilbert space decomposition into
invariant parity subspaces so no estimates of the type of Theorem \ref{thm:asymptotic} are known.
Actually, as mentioned above, there is no known Hilbert space decomposition into invariant subspaces independently from $g, \Delta$ even in the case where  $\e$  is half-integer  \cite{A2020}, simply because $J_\ell$ depends explicitly on $g$ and $\Delta$ and 
\[
(g,\Delta)\not= (g',\Delta') \; \Rightarrow \;[J_\ell(g,\Delta),J_\ell(g',\Delta')]\neq 0.  
\] 

We will show now that there is nevertheless an analog of Proposition \ref{prop:measureParity} for the AQRM. First,
for $N \in \mathbb{Z}_{\geq 0}$ we define 
\begin{align}
  \alpha_0^{(\e)}(N) = \alpha_0^{(\e)}(N;g,\Delta)& :=\max_{n=1,2,\ldots,N-1} \{\lambda_{n+1}-\lambda_n | \lambda_n,\lambda_{n+1} \in \Spec(\AHRabi{\e}) \},\nonumber\\
  M^{(\e)}_N(\alpha) =  M^{(\e)}_N(\alpha:g,\Delta)&:=\# \{n\leq N \,|\, 0\leq \lambda_{n+1}-\lambda_n<\alpha\}, \quad \alpha\leq \alpha_0^{(\e)}(N).\nonumber
\end{align}
  
For the QRM, we obtain the following measure for the full spectrum, comprising both parities:

\begin{prop}
  \label{prop:measureQRM}
  There is a measure $\mu$ in $[0,\infty)$ such that
  \[
    \lim_{N\to \infty}\frac{M^{(0)}_N(\alpha)}{M^{(0)}_N(\alpha_0^{(0)}(N))}= \int_0^\alpha \mu_0(x) dx
  \]
  where $\alpha_0 = {\rm sup}(S(g,\Delta,0))$. Here, $\mu_0(x)$ is the measure given by
  \[
    \mu_0(x) = \frac12 \delta(x)+ \frac12 \delta(x-1).
  \]
\end{prop}

\begin{proof}
  Let $\eta>0$, then by Theorem \ref{thm:asymptotic} above, there is an $N \in \Z_{\geq 0}$ such that for $n>N$, we have
  \[
    \lambda_{n+1}- \lambda_n \in (0,\eta) \cup (1 -  \eta, \, 1 + \eta) 
  \]    
  and the same argument used in the proof of Proposition \ref{prop:measureParity} shows that
    \[
    \mu_0(x) = c \delta(x) + (1-c)\delta(x-1),
  \]
  for some constant $0<c <1$. To show that $c = \frac12$ it is suffices to show that for large $n$,
  the energy spacing alternates between the intervals $(0,\eta)$ and $(1 -  \eta, \, 1 + \eta)$ .

  To see that for $n>N$, let us first assume that
  \[
  \lambda_{n+1}- \lambda_n \in (0,\eta).
  \] 
  Then, it follows from Theorem \eqref{thm:asymptotic} that $\lambda_{n+1}$ and $\lambda_n$ are of different parity. Namely, there exists $m\in\N$ such that $\lambda_{n+1}=\lambda_m(H_\pm)$ and $\lambda_{n}=\lambda_m(H_\mp)$. This implies that $\lambda_{n+2}=\lambda_{m+1}(H_+)$ or $\lambda_{n+2}=\lambda_{m+1}(H_-)$. It follows, in particular that $\lambda_{n+2}- \lambda_{n+1} \in (1-\eta, 1+\eta)$.
  
  Next we assume that 
  \[
  \lambda_{n+1}- \lambda_n \in (1-\eta, 1+\eta).
  \] 
  Then, there exists $m\in\N$ such that either of the following cases hold:
  \begin{itemize}
  \item $\lambda_{n+1}=\lambda_{m+1}(H_\pm)$ and $\lambda_{n}= \lambda_{m}(H_\pm)$.
   \item $\lambda_{n+1}=\lambda_{m+1}(H_\pm)$ and $\lambda_{n}= \lambda_{m}(H_\mp)$.
  \end{itemize}
  For each case, we have $\lambda_{n+2}=\lambda_{m+1}(H_\mp)$.  Hence,  it follows that $\lambda_{n+2}- \lambda_{n+1} \in (0, \eta)$.
  This completes the proof. 
  \end{proof}
  
\begin{cor}
  \label{cor:Alternating}
  For a sufficiently small $\eta>0$ define 
  $$
  D_\eta(N; g,\Delta): =  \frac{ \# \{n\leq N \, | \, \lambda_{n+1}-\lambda_{n}  \in (0,\eta)\}}
  {\# \{n\leq N \, | \, \lambda_{n+1}-\lambda_n \in (1-\eta, 1+\eta)\}}.
$$
Then we have $\lim_{N\to\infty}D_\eta(N; g,\Delta)=1$. 
\end{cor}
\begin{proof}

Since the energy spacing alternates between the intervals $(0,\eta)$ and $(1 -  \eta, \, 1 + \eta)$ as observed in the proof of Proposition  \ref{prop:measureQRM}, the energy spacing is equidistributed in the two intervals. 
\end{proof}

Proposition \ref{prop:measureParity} suggests that the limiting distribution (measure) for the level spacing does not
depend on the parameters $g,\Delta>0$. To consider the effect of the bias parameter $\e$, in Figure \ref{fig:ES5} we show
the scatter plot of $S(g,\Delta,\e)$ for different bias parameters $\epsilon$ but with other parameters unchanged.

\begin{figure}[ht!]
\begin{center}
    \includegraphics[scale=0.38]
    {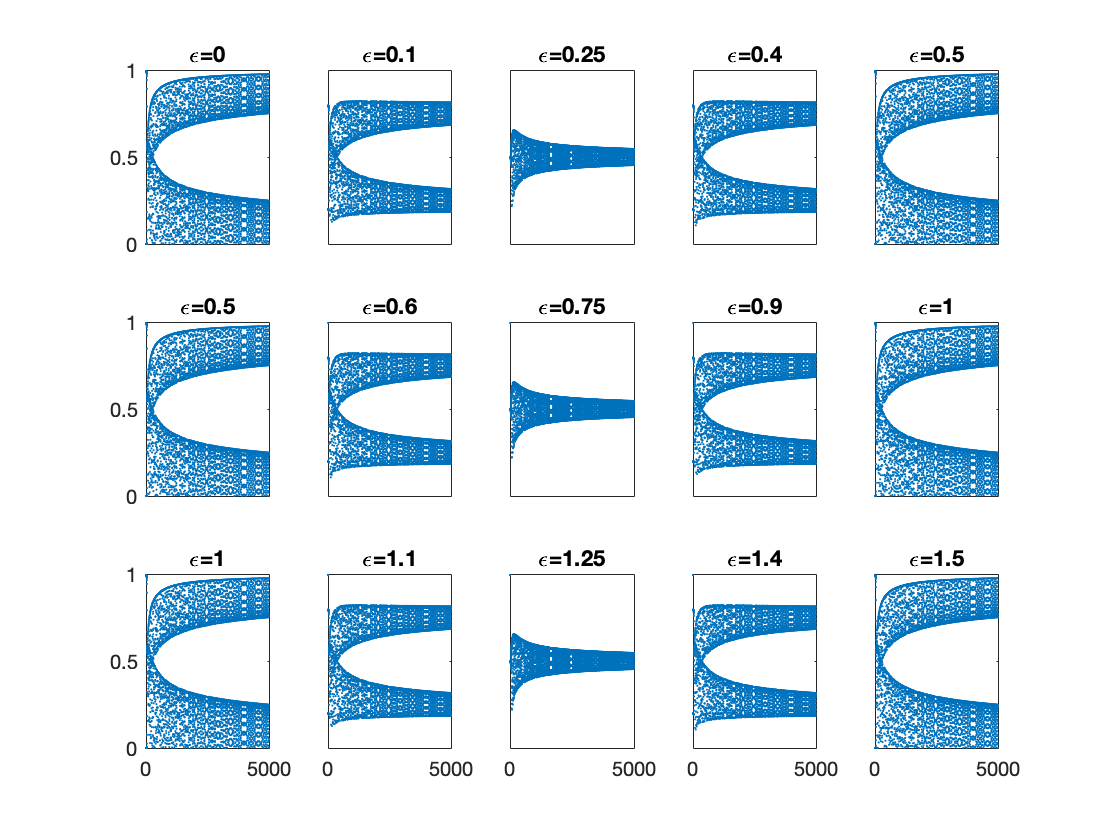}
\end{center}
    \caption{Scatter plot of energy spacing AQRM with $g=5$, $\Delta=5$ and $\epsilon$ is tuned from $0$ to $1.5$}
    \label{fig:ES5}
\end{figure}

We see that the spacing distribution for all $\epsilon \in \frac12 \Z$ shows the same pattern. This is expected since the study of this case (e.g. \cite{KRW2017}) has shown that with the exception of a small number of low-lying eigenvalues, the spectral structure (including the existence of a symmetry operator) is similar to the QRM. But we see that the periodicity in $\e$ with period $1/2$ extends to the general case $\epsilon \notin \frac12 \Z$. That is, the general shape of the energy spacing graphs corresponding to  the models with bias $\epsilon + \frac{n}{2}$ are similar for any $n \in \Z$.  The numerical results obtained here show that even in the general case, the spectral structure, at least for high level numbers, is similar for all bias parameters of the form $\epsilon + \frac{n}{2}$. We further discuss this periodicity in Section \ref{sec:refinement}. We remark that we verified that the general shape of the graphs and behavior observed in Figure~\ref{fig:ES5} holds for a number of choices of parameters $g, \Delta$.

To consider the maximum possible energy spacing, we define
\[
  \alpha_0(g,\Delta,\e) := {\rm sup} \, S(g,\Delta,\e).
\]
We have shown that $\alpha_0(g,\Delta,0)< \infty$ and the numerical experiments suggests that actually $\alpha_0(g,\Delta,\e) \leq 1$. Moreover, from numerical experiments it appears that $\alpha_0(g,\Delta,\e)\leq 1$ for any set of parameters.

It is important to notice that for fixed parity (if parity exists, i.e. $\e \in \frac12 \Z$), the spacing may be larger than $1$. Since
$\alpha_0(g,\Delta,0) \leq 1$, it follows that  whenever we have
\[
  \lambda_{n+1}(H_\pm) - \lambda_n(H_\pm) > 1,
\]
then there must be an eigenvalue $\lambda_m(H_\mp)$ such that
\begin{align}
  \label{eq:mixedcrossing}
  \lambda_{n}(H_\pm) \leq \lambda_m(H_\mp) \leq  \lambda_{n+1}(H_\pm).
\end{align}
Note that by Theorem \ref{thm:asymptotic}, when $n$ is sufficiently large, $m$ should be equal to $n$ or $n+1$. Here, the value of $m$ may also depends on the parity (even or odd) of $n$. This discussion is, of course, consistent with the Weyl law \cite{RW2021} for each parity (see also \cite{R2020}). This further shows that even for small eigenvalues (i.e., even for a finite number of small integers $n$), there is not much variation in the possible values of number $m$ once $n$ is given. This fact may also support the G-conjecture in \cite{Braak2019} (see Subsection \ref{sec:remarksG}).

In fact, note that for fixed $\delta>0$ and large enough $n$,
\[  
    \lambda_{n+1}(H_+) - \lambda_{n}(H_+) + \lambda_{n+1}(H_-) - \lambda_{n}(H_-) = 2 +  O(n^{-1/2+\delta}) 
\] 
so we cannot determine the order of the eigenvalues and the type of the spacing only from the asymptotic estimates.

Next, for the half-integer bias case, we consider the parities to make a more detailed analysis of the energy
distribution. However, due to the considerations in Section \ref{sec2:EigenvalueComputation}, we consider only the symmetric case here $\e=0$.  We note here that the ground state is simple and has positive parity \cite{HH2012,HHL2014}.

\begin{dfn}
  \label{dfn:spacingClassification}
  For $\e \in \frac12 \Z$. We say that the energy spacing
  \(
    \lambda_n - \lambda_{n-1} \in S(g,\Delta,\e)
  \)
  is of
  \begin{enumerate}
  \item {\em positive type} if $\lambda_n,\lambda_{n-1} \in H_+$,
  \item {\em negative type} if  $\lambda_n,\lambda_{n-1} \in H_-$,
  \item {\em mixed type} in any other case.
  \end{enumerate}
  In other words, $S(g,\Delta,\e)$ is the disjoint union
  \[
    S(g,\Delta,\e) = S_+(g,\Delta,\e) \cup S_-(g,\Delta,\e) \cup S_m(g,\Delta,\e),
  \]
  where each of the sets $S_\pm(g,\Delta,\e)$ and $S_m(g,\Delta,\e)$  correspond to the energy spacing of positive, negative, and mixed types.
\end{dfn}

\begin{figure}[h!]
  \begin{center}
\includegraphics[scale=0.18]{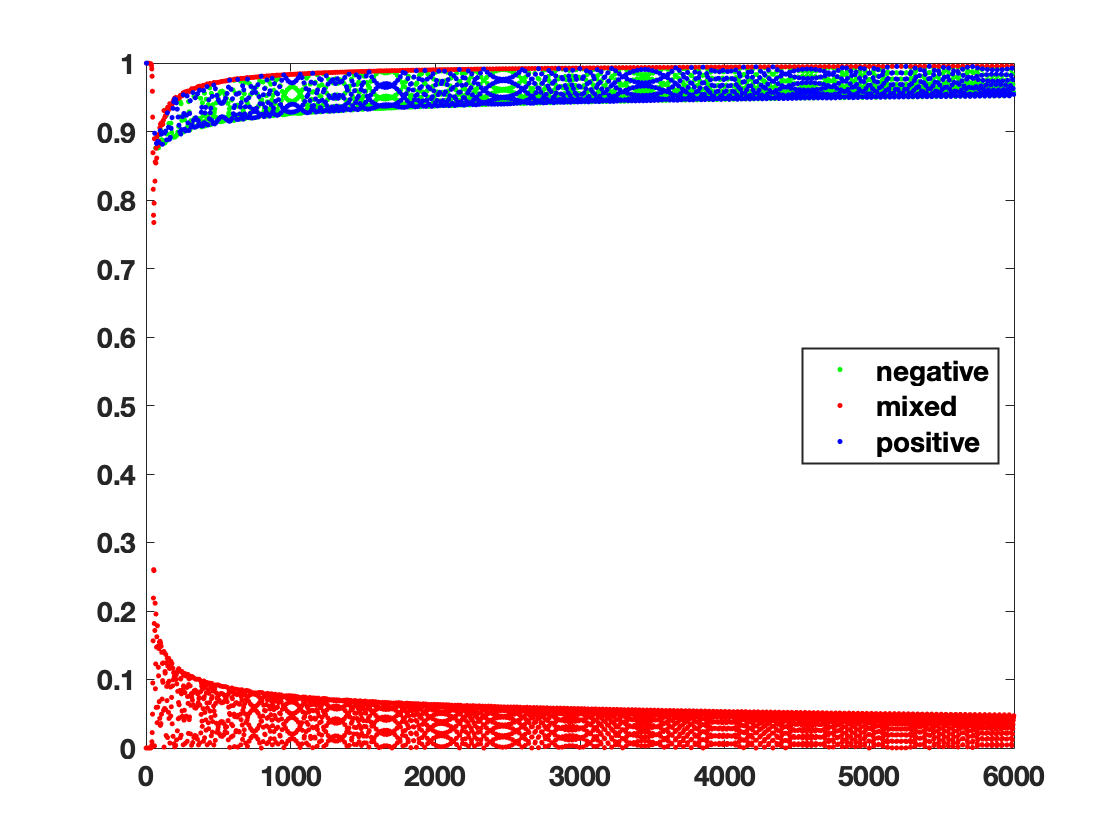}\\
\caption{Energy spacing for the QRM with $g=5, \Delta=1$. Consecutive energy spacing of positive type (blue), negative type (green), and mixed type (red).}
\label{fig:Distinguish_ES_on_Parities}
\end{center}
\end{figure}

In Figure \ref{fig:Distinguish_ES_on_Parities} we show the energy spacing for a specific QRM with points colored according to the classification of Definition \ref{dfn:spacingClassification} for a finite number of energy points. We observe that the number of spacing points of positive and negative type are approximately equal. Also, we see that the number of mixed type spacing points appear to exceed the sum of positive and negative type spacing points. Clearly, at least one of the sets $S_\pm(g,\Delta,\e)$ and $S_m(g,\Delta,\e)$ must be infinite, so in order to describe the distribution of spacing points of each type  as $N \to \infty$  we define
\begin{equation}
  \label{eq:density1}
 d(N;g,\Delta) := \frac{ \# \{n\leq N \, | \, \lambda_n,\lambda_{n-1} \in H_{+} \text{ or } \lambda_n,\lambda_{n-1} \in H_{-} \} }{N},
\end{equation}
and
\begin{equation}
  \label{eq:density2}
   r(N;g,\Delta) :=  \frac{ \# \{n\leq N \, | \, \lambda_n,\lambda_{n-1} \in H_{-}\}}{\# \{n\leq N \, | \, \lambda_n,\lambda_{n-1} \in H_{+}\}}.
\end{equation}

To understand the effect of the system parameters on the functions \eqref{eq:density1} and \eqref{eq:density2}, we compute them for different parameter combinations. Figure~\ref{fig:energy_spacing_type_proportion} shows the ratios when $N$ increases. We note that the density function depends on the value of $g$ and that the parameter $\Delta$ does not have a significant effect. Concretely, we observe that $d(N;g,\Delta)$ is a monotonically increasing function of $N$. Moreover, when we fix parameter $N$, this function is decreasing with respect to $g$. 

We can approximate the function $d(N;g,\Delta)$ by 
\begin{equation}
  \label{eq:es_type_proportion_fitting_function}
  \frac{\text{\rm {arctan}}(a \log_{10}(N)+ b)}{\pi} +c,
\end{equation}
where the parameters $a=a(g,\Delta), b=b(g,\Delta)$ can be estimated from the calculated data via the curve fitting method.
Their values for three chosen models, all with $\Delta=1$, are given in Table~\ref{tab:energy_spacing_type_proportion_fitting_result}. To access the goodness of the fit, we use the sum squared error (SSE), coefficient of determination ($R$-square), and root-mean-square error (RMSE). We observe that the ratios tend to $0.5$ as $N$ approaches infinity for all three choices of parameters.

\begin{figure}[h!]
\begin{center}
\includegraphics[scale=0.2]{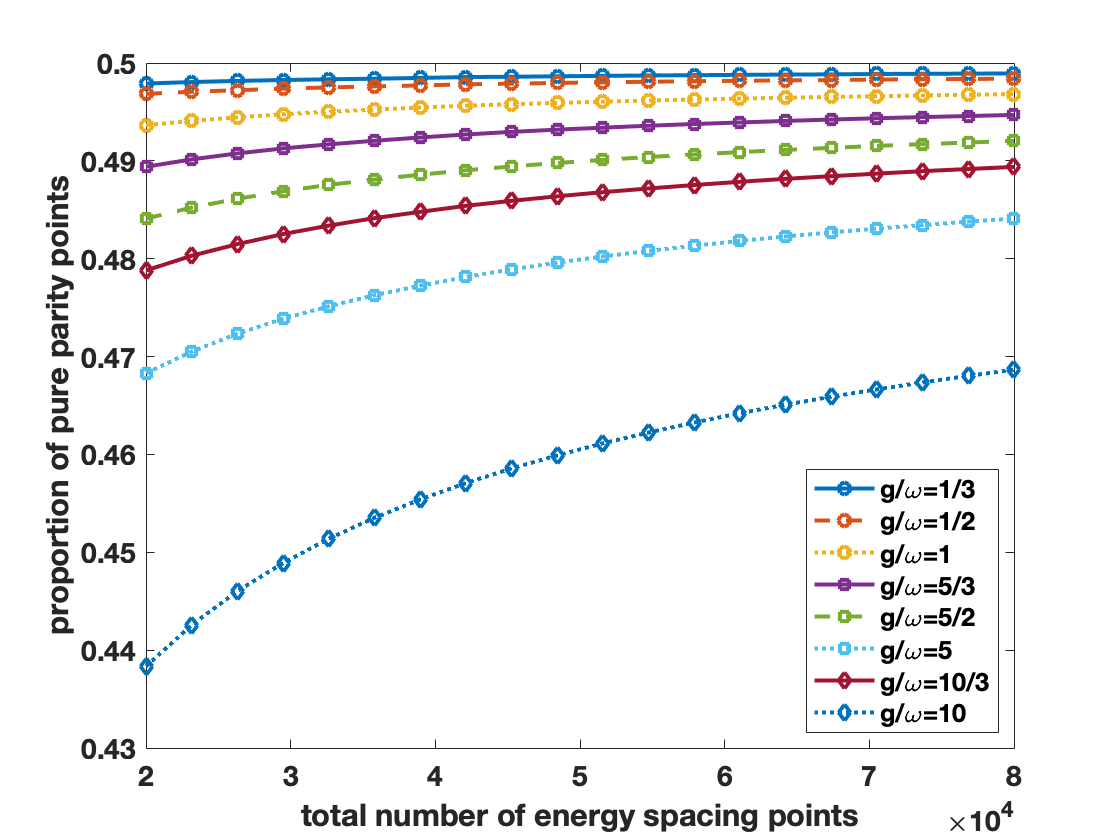}
\caption{Proportion of the number of adjacent levels with the same parity to the total number of levels.}
\label{fig:energy_spacing_type_proportion}
\end{center}
\end{figure}

\begin{table}[!ht]
\centering
      \begin{tabular}{cccccccc}
        \hline
        $g$ &$a$&$b$&$c$&SSE&$R$-square&RMSE\\ %&$\lim_{N\to\infty}p(N)$\\
        \hline
        $1/3$&$21.67$&$-42.15$&$0.004$&$9.602\times 10^{-8}$&$0.9458$&$7.515 \times 10^{-5}$ \\ %& $0.5040$\\
        $5/2$&$31.46$&$-115.9$&0&$9.773\times 10^{-7}$&$0.9905$&$0.000233$\\ %&$0.5$\\
$10$&$7.922$&$-29.12$&$0$&$1.374\times 10^{-5}$&$0.9909$&$0.0008735$\\ %&$0.5$
      \end{tabular}
 \caption{Parameters of the fits according to formula \eqref{eq:es_type_proportion_fitting_function} for $\Delta=1$ and various values of $g$.}
\label{tab:energy_spacing_type_proportion_fitting_result}
\end{table}

From the numerical data and the asymptotic behavior of the fitting function, we expect that
\begin{equation}
    \label{eq:densityExpect1}
    \lim_{N\to \infty} d(N;g,\Delta) = \frac12, \qquad \lim_{N\to \infty} r(N;g,\Delta) = 1
\end{equation}
hold independent of the parameters. Namely, the fraction of level spacing of positive or negative type is
the same as that of mixed type. Similarly, the fraction of energy spacing of  positive type is the same as that of
negative type. In particular, all three sets $S_\pm(g,\Delta,\e)$ and $S_m(g,\Delta,\e)$ are infinite. The results \eqref{eq:densityExpect1} were proven mathematically by Rudnick in \cite{R2023} as a corollary of a density one version of the G-conjecture (cf. Section \ref{sec:refinement}).

We also note here that the two limits in \ref{eq:densityExpect1} are not equivalent, in other words, neither one of the two necessarily implies the other. In Section \ref{sec:refinement}, we consider the general case for the AQRM based on the numerical results in this section.

\begin{rem}
  The density result of Corollary \ref{cor:Alternating} should not be confused with the density described by the ratio described  by the function $d(N;g,\Delta)$ and the corresponding density at the expected limit
  \[
    \lim_{N\to \infty} d(N;g,\Delta) = \frac12.
  \]
  In particular, despite the similarity of both densities when observed in the numerical graph for the QRM in Figure \ref{fig:Distinguish_ES_on_Parities}, we remark the result of Corollary \ref{cor:Alternating} does not consider the type (parity) of eigenvalue (see the discussion in the proof of Proposition  \ref{prop:measureQRM}).
\end{rem}

\begin{rem}
  The $\frac14$-symmetry may be partially explained by the pole structure of the $G$-function
  $G^{(\e)}(x;g,\Delta)$ \cite{KRW2017} of the AQRM. The main property of the $G$-function is that if
\[
    G^{(\e)}(x;g,\Delta) = 0,
  \]
then $x+g^2$ is an eigenvalue of the AQRM for the parameter $g$ and $\Delta$.  It is well known that the $G$-function has poles at the points $x= N \pm \e$ (see e.g., \cite{KRW2017}) and when there is a Juddian eigenvalue these singularities are ``lifted'' (see \cite{K1985JMP} for the QRM case and \cite{LB2015JPA,KRW2017} for the general AQRM case). We illustrate this situation in Figure \ref{fig:gfunct_graph3}. Note that the distances between consecutive poles are of length $2\{\e\}$ and $1-2\{\e\}$, so that they are the same for $\e=\frac18$ and $\e=\frac38$, and both distances coincide only for
$\e \in \frac14 + \frac12 \Z$. 

\begin{figure}[htb]
  ~
  \begin{subfigure}[b]{0.45\textwidth}
    \centering
    \includegraphics[height=3.25cm]{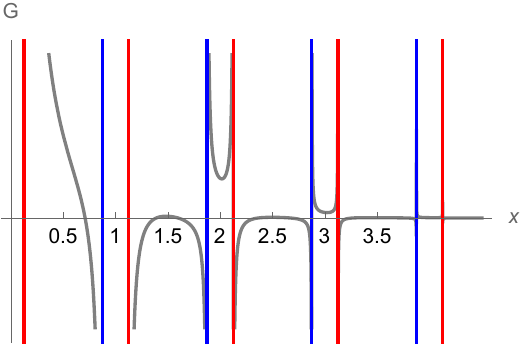}
    \caption{\(\e = \tfrac18\)}
  \end{subfigure}
  ~
  \begin{subfigure}[b]{0.45\textwidth}
    \centering
    \includegraphics[height=3.25cm]{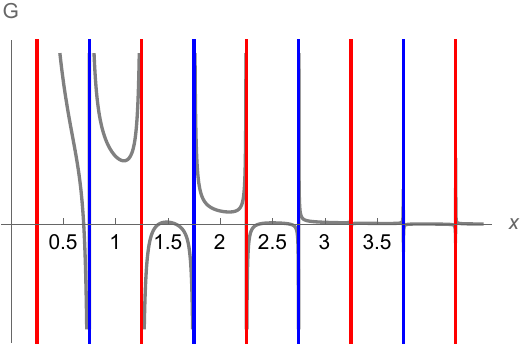}
    \caption{\(\e = \tfrac14\)}
  \end{subfigure}
  \\
  \centering
  \begin{subfigure}[b]{0.45\textwidth}
    \centering
    \includegraphics[height=3.25cm]{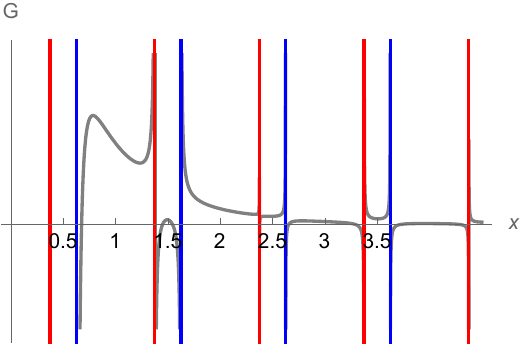}
    \caption{\(\e = \tfrac38\)}
  \end{subfigure}
  ~
    \begin{subfigure}[b]{0.45\textwidth}
    \centering
    \includegraphics[height=3.25cm]{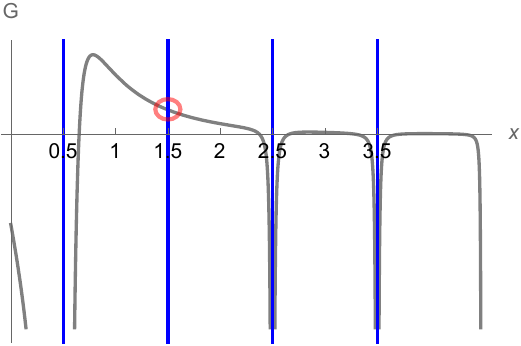}
    \caption{\(\e = \tfrac12\)}
  \end{subfigure}
  \caption{Plot of \(G^{(\e)}(x;g,\Delta)\) for parameters \(g=\tfrac12\) and \(\Delta=1\) for different $\e$ with poles lines indicated $n+\e$(red) and $n-\e$ (blue). Note the ``lifted'' Juddian eigenvalue for the case $\e=\tfrac12$ (red circle).}
  \label{fig:gfunct_graph3}
\end{figure}

\end{rem}

\subsection{Density functions for the spacing distribution}
\label{sec:density}

In this section, we introduce numerical density functions that approximate the actual distribution of the energy spacing for the full spectrum 
of the AQRM.

For every $N>0$, define a partition of $[0,\alpha_0^{(\e)}(N)]$ of the form
\[
    0=\alpha_1<\alpha_2<\ldots<\alpha_{L-1}<\alpha_L=\alpha_0^{(\e)}(N),
\]
then, the numerical density functions $f_N^{(\e)}(\alpha_l;g,\Delta)$ is  defined by
\begin{equation*} 
    f_N^{(\e)}(\alpha_l;g,\Delta)= \dfrac{M^{(\e)}_N(\alpha_{l+1})-M^{(\e)}_N(\alpha_{l})}{(\alpha_{l+1}-\alpha_{l})N},
\end{equation*}
for $l=1,2,\ldots,L-1$.

\begin{figure}[!ht]
\begin{center}
    \includegraphics[scale=0.19]{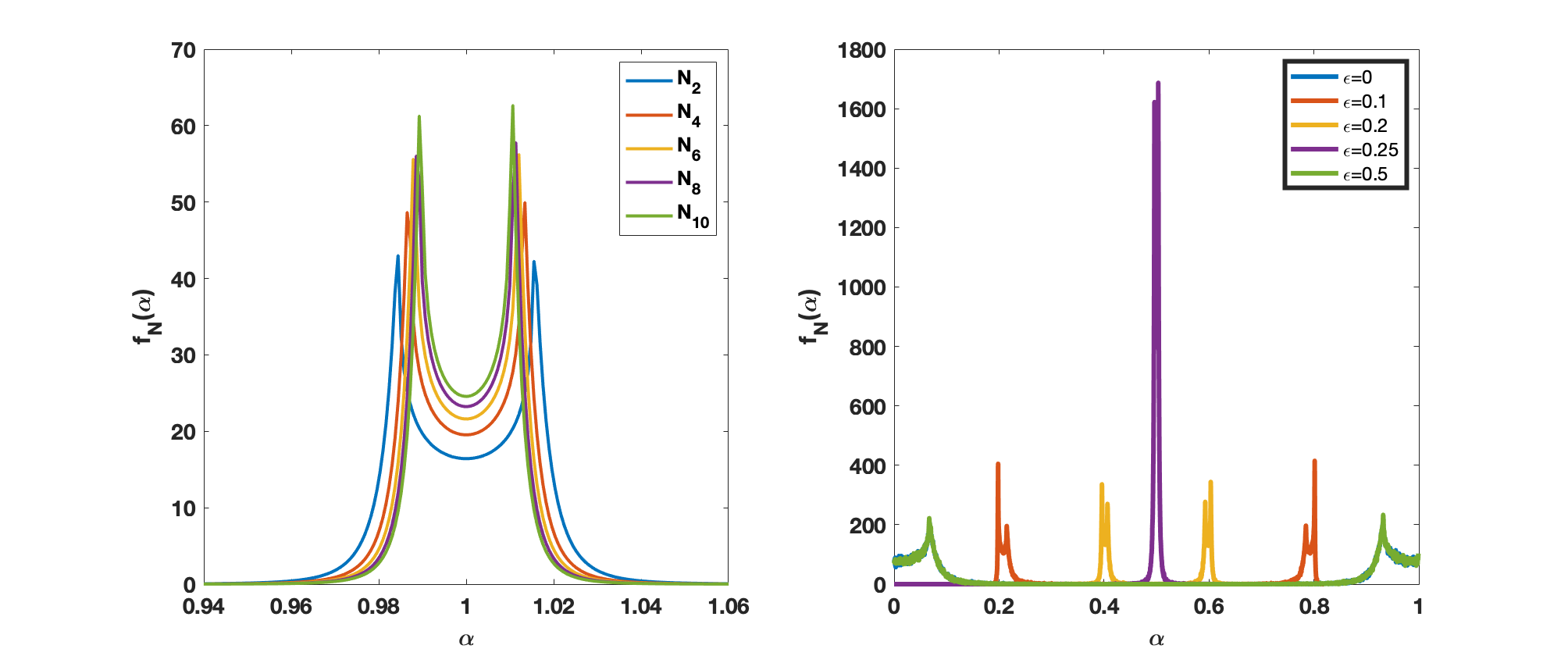}
  \end{center}
      \caption{(Left) Density of energy spacing for positive parity of the QRM as a function of truncation threshold $N$ with $N_k = 3.5 k \times 10^6$; 
    (Right) The density of energy spacing for QRM and AQRM with different bias values.} 
\label{fig:truncated_density_vs_N}
\end{figure}

As an example, let us consider the case of fixed parity within the QRM. The left panel of Figure \ref{fig:truncated_density_vs_N} shows the graph of the numerical density of energy spacing for positive parity for $g=\Delta=1$. Here $N= N_0 \times i$, where $i=1,2,...,10$ and $N_0 = 3,500,000$. The first $40,000$ eigenvalues are calculated using exact diagonalization, and the larger eigenvalues are approximated using \eqref{eq:asymptotics}. In particular, this illustrates the approximation of the delta measure in the proof of Proposition \ref{prop:measureParity}.

The right panel in Figure \ref{fig:truncated_density_vs_N} compares the numerical density function of AQRM with different bias values. We note that the asymptotic periodicity discussed in the previous section is reflected here, as illustrated in Figure \ref{fig:density_function_wrt_epsilon_second_parameter}. In the right panel, the first $40,000$ eigenvalue spacings are collected into $2,000$ subdivisions of $[0,1]$ of equal length.

\begin{figure}[!ht]
    \begin{center}
        \includegraphics[scale=0.31]{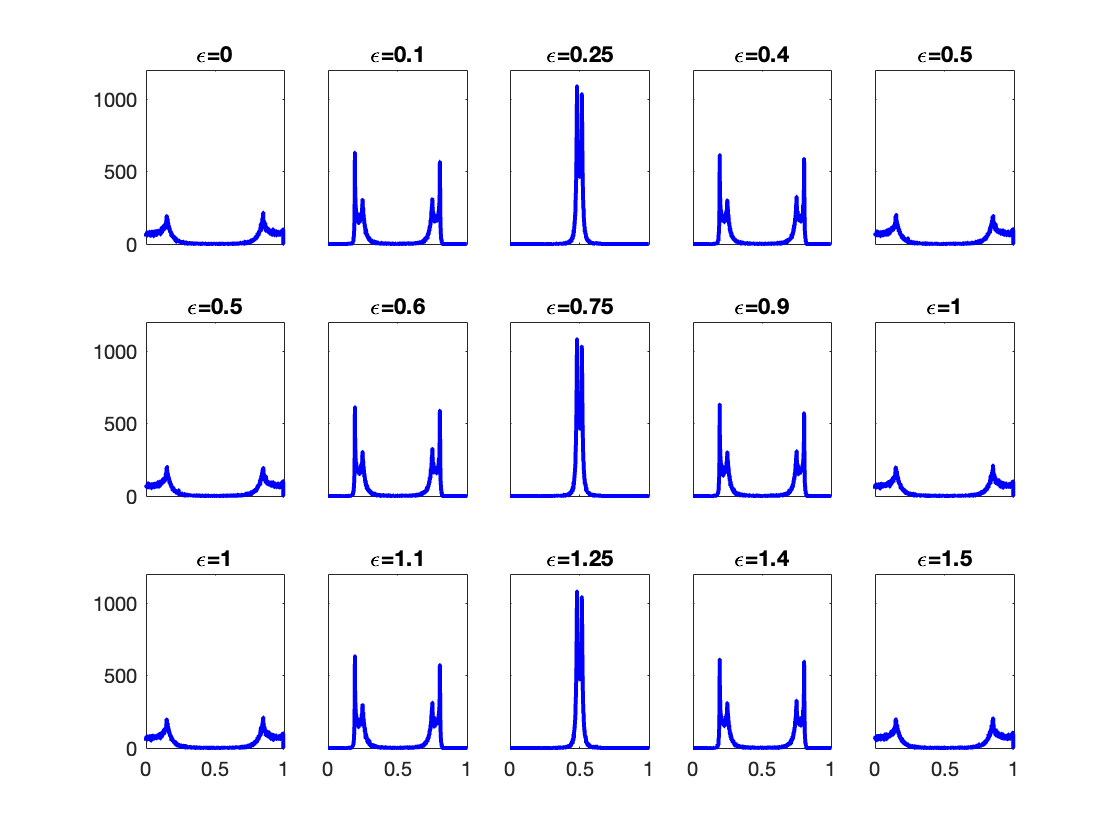}
        \caption{Numerical density functions with respect to changing $\epsilon$ ($g=5$, $\Delta=5$)}
    \label{fig:density_function_wrt_epsilon_second_parameter}
    \end{center}
\end{figure}

From the structure of the AQRM for half-integer bias and the $\frac12$-periodicity observed before, 
the measure $\mu_{\frac{\ell}2}(x)$ of the energy spacing is expected to have the form
\[
   \mu_{\frac{\ell}2}(x) = c_1 \delta(x)+ c_2 \delta(x-1)
\]
for real constants  $c_1$, $c_2$ with $c_1+c_2=1$. The general case and the values of the constants
are discussed in the next sections, along with results and conjectures on the eigenvalue distribution for
the AQRM.

We conclude this subsection by noting another symmetry with respect to $\e=\frac{n}{4}$ for any fixed natural number
$n$. Figure \ref{fig:density_function_wrt_epsilon_second_parameter} shows that the energy spacing distribution for $\epsilon = \frac{n}{4} + \eta $ and $\epsilon = \frac{n}{4} - \eta $ for  $0 \leq \eta < \frac14$ resemble each other. To the best of our knowledge, this type of symmetry has not been previously observed. In the next section, we describe with more detail the symmetric behavior of the energy spacing under certain assumptions.

\subsection{Numerical cumulative distribution}
\label{sec:cumulative}

In this section we consider numerically the cumulative distribution of energy spacing of the AQRM to attempt to understand the density measure of the AQRM for general $\epsilon \neq 0$. From the numerical observations of Section \ref{sec:density}, 
it is natural to expect that the density measure is given as a sum of delta distributions centered around certain critical values, similar to the  QRM case. In addition to finding the critical values, the use of the numerical cumulative distribution provides an analogy to the hyperbolic Laplacian  discussed in the introduction (cf. equation \eqref{eq:Laplace}).

Accordingly, we define $h_\e(N;\alpha)$ by
\[
  h_\e(N;\alpha) = \frac{M^{(\e)}_N(\alpha)}{N},
\]
to study the behavior of $\lim_{N \to \infty} h_\e(N;\alpha)$ as $N \to \infty$  and $0 \leq \alpha \leq \alpha_0^{(\e)}$.  

From the numerical computations, we identify that $h_\epsilon(N;\alpha)$ has two critical points
\[
  \alpha_1 = \min(\{2 \epsilon\},1-\{2 \epsilon\}), \qquad \alpha_2 =  \max(\{2 \epsilon\},1-\{2 \epsilon\}),
\]
corresponding to the critical values at the peaks of Figure \ref{fig:density_function_wrt_epsilon_second_parameter} for $\e \not\in \frac12\Z$ in  Section \ref{sec:density}. Here, $\{x\}$ denotes the fractional part of $x\in\R$.
Note that when $\e \in \frac14 + \frac12 \Z$, the two critical values coincide, that is $\alpha_1 = \alpha_2$, and the function has a unique critical value at $\alpha_1 = \frac12$. The plot of $h_\epsilon(N;\alpha)$ for different values of $\alpha$ around the smaller critical value are given in Figures \ref{fig:Heps02} and \ref{fig:Heps06} for $\epsilon=0.2$ and $\epsilon=0.6$, respectively.

\begin{figure}[!ht]
\centering
\includegraphics[width=\textwidth]{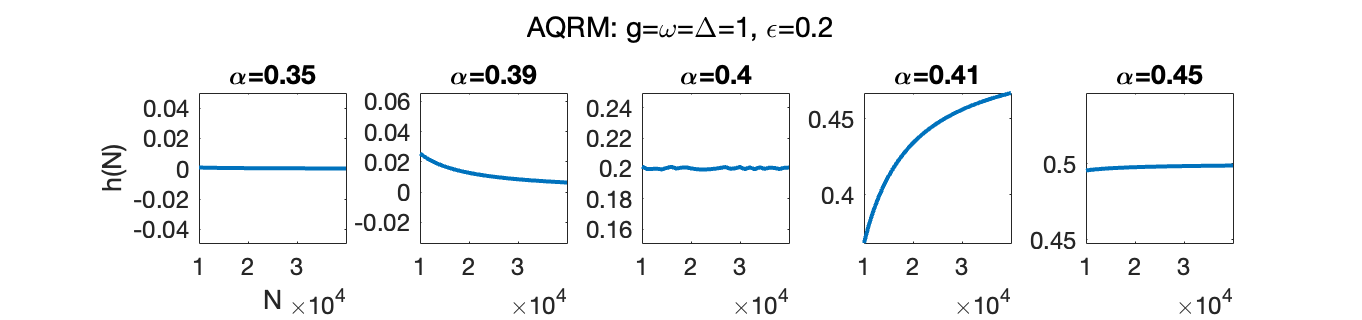}
\caption{Numerical value of $h_\epsilon(N;\alpha)$ for  $\epsilon=0.2$ around the critical value $\alpha_1=0.4$ for $N \leq 40,000$.}
\label{fig:Heps02}
\end{figure}

\begin{figure}[!ht]
\centering
\includegraphics[width=\textwidth]
%{figures/Heps06.png}
{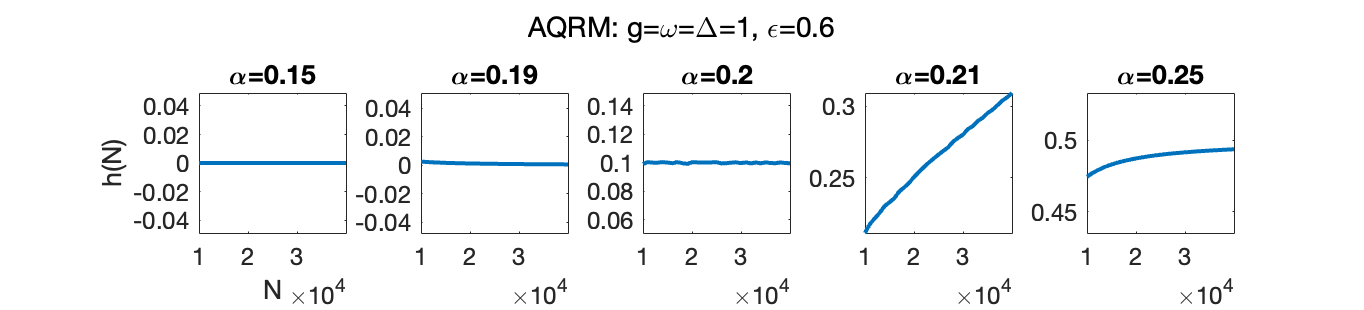}
\caption{Numerical value of $h_\epsilon(N;\alpha)$ for $\epsilon=0.6$ around the critical value $\alpha=0.2$ for $N \leq 40,000$.}
\label{fig:Heps06}
\end{figure}

Next, let us describe the behavior of $h_\epsilon(N;\alpha)$ with respect to $N$ about $\e (\notin \frac12 \Z)$ around the first critical value $\alpha_1$ from Figures \ref{fig:Heps02}  and \ref{fig:Heps06}. We only consider the first critical value $\alpha_1$ since a similar behavior is observed around the second critical value $\alpha_2$.

We observe that
\begin{itemize}
    \item for $\alpha < \alpha_1  $,  $h_\epsilon(N;\alpha)$ decreases to $0$ rapidly (with speed depending on the distance $\alpha_1- \alpha$) with respect to $N$,
    \item at the critical point $\alpha=\alpha_1$, the function $h_\epsilon(N;\alpha)$ oscillates around a constant value $c_1(\e)$,
    \item for $\alpha > \alpha_1$ the function $h_\epsilon(N;\alpha)$ increases with $N$ to the value $\frac12$. 
\end{itemize}
Interestingly, we note that in general the rate of increase for the $\alpha > \alpha_1$ case is not as pronounced as that of the decrease of the $\alpha < \alpha_1$ case.

The numerical results suggest that the limiting case of the cumulative distribution is given by
\[
    \lim_{N\to \infty} h_\epsilon(N;\alpha)= 
    \begin{cases}
      0   & \text{ for } \alpha < \alpha_1 \\
      c_1(\e)   & \text{ for } \alpha = \alpha_1 \\
      \frac12  & \text{ for }  \alpha_1  < \alpha < \alpha_2 \\
      c_2(\e)   & \text{ for } \alpha = \alpha_2 \\
      1   & \text{ for } \alpha_2  < \alpha
    \end{cases},
  \]    
for some constants $(0 \leq) c_1(\e), c_2(\e) (\leq 1)$. An obvious modification is needed for the case $\epsilon \in \frac14 + \frac12 \Z$ with a unique critical value.

Based on this numerical observations and recent asymptotic results, in the following section we describe the density distribution of energy spacing of the general case of the AQRM. We conclude the section by noting that while the actual values of $c_i(\e)$ are not relevant for the distribution measure of energy spacing, for all computed parameters the values of $c_i(\e)$ are given by
\[
    c_1(\e) = \frac12 \alpha_1, \qquad c_2(\e) = 1 - \frac12 \alpha_1,
\]
and by $c_1(\e) = \frac12$ for $\e \in \frac14 + \frac12 \Z$.

\subsection{Periodicity and symmetry of energy spacing with respect to the bias parameter}
\label{sec:refinement}

Based on the numerical observations of the previous sections and the recent result on the asymptotics of the eigenvalues for the AQRM \cite{cfz2023},  we describe the distribution for the spacing of  the AQRM in general.

\begin{thm}
 \label{conj:measure}
  For any $\e \in \R$, there is a measure $\mu_\e$ supported in $[0,\infty)$ such that
  \[
    \lim_{N\to \infty}\frac{M^{(\e)}_N(\alpha)}{M^{(\e)}_N(\alpha_0^{(\e)})}= \int_0^\alpha \mu_\e(x) dx.
  \]
  where $a_0^{(\e)}= {\rm sup} \left(S(g,\Delta,\e)\right)< \infty$.
The measure is given by
  \begin{equation}
    \label{eq:measure}
    \mu_\e(x) =  \frac12 \delta(x- \{2\epsilon\}) + \frac12 \delta(x -(1- \{2\epsilon\})).
  \end{equation}
%  where $\{x\}$ denotes the fractional part of $x \in \R$. 
\end{thm}

\begin{proof}
    The proof follows in a similar way to Proposition \ref{prop:measureParity} by using the asymptotic estimate
    \[
        \lambda_n^{\pm} = n - g^2 \pm \frac12\e + O(n^{-\frac14}),
    \]
    for the eigenvalues $\lambda_n^{\pm}$ of the AQRM proved by Charif, Fino and Zielinski in \cite{cfz2023}. We note that the asymptotic estimate does not have an oscillatory term, but this is not a limitation for the proof.
\end{proof}

Additionally, according to the numerical computations, we conjecture that \( \alpha_0^{(\e)} < 1\). That is, the energy spacing never exceeds $1$.

%In addition to the numerical evidence, it is important to notice that this conjecture actually holds in the weak
%limit $g \to \infty$ where the spectrum is known exactly (see \eqref{eq:limitspec}).

The following corollary proves the periodicity and symmetry for the
family of the AQRM  discussed in the previous section.

\begin{cor}
  \label{cor:sym}
  By Theorem \ref{conj:measure}, for $n \in \Z$ we have
  \[
    \mu_{\e}(x) = \mu_{\e+ \frac{n}2}(x), \qquad (\tfrac12-\text{periodicity})
  \]
  and
  \[
    \mu_{\frac{n}4 + \e}(x) = \mu_{\frac{n}4 - \e}(x), \qquad  (\tfrac14-\text{symmetry}).
  \]
\end{cor}

\begin{proof}
The proof follows directly from \eqref{eq:measure} by using the elementary identity
  \[
    \{- 2\e\} = 1 - \{ 2\e\}.
  \]
\end{proof}

Following the discussion of Section \ref{sec:preliminaries}, for the case $\e = \frac{\ell}2 \in \frac12 \Z$ we define
\begin{equation}
  \label{eq:density3}
  d_{\ell}(N;g,\Delta) := \frac{ \# \{n\leq N \, | \, \lambda_n,\lambda_{n-1} \text{ are both of the same parity}\}} {N}
\end{equation}
and
\begin{equation}
  \label{eq:density4}
   r_\ell(N;g,\Delta) :=  \frac{ \# \{n\leq N \, | \, \lambda_n,\lambda_{n-1} \text{ have negative parity}\}}{\# \{n\leq N \, | \, \lambda_n,\lambda_{n-1}  \text{ have positive parity}\}}.
\end{equation}
Note that this is well defined with the possible exception of a finite number of eigenvalues. We expect the density of the parities defined in this way to follow the same pattern as in the QRM.

\begin{conj}
  For any $\ell \in \Z$, we have
\[
  \lim_{N\to \infty} d_\ell(N;g,\Delta) = \frac12, \qquad \lim_{N\to \infty} r_\ell(N;g,\Delta) = 1.
\]
\end{conj}

To conclude this section, we note that by Theorem \ref{conj:measure}, for the case $\e=\frac14$ the distribution is 
given by
\[
  \mu_{\tfrac{n}4}(x) = \delta(x-\tfrac12),
\]
resulting in a spacing distribution similar to the case of a single parity of the QRM, which may be related to the
symmetry in Corollary \ref{cor:sym}.

\subsection{Remarks on the G-conjecture}
\label{sec:remarksG}

The $G$-conjecture is a conjecture on the eigenvalue distribution of the QRM (see \cite{Braak2019}) that remains open despite being supported by extensive numerical studies and asymptotic methods (see e.g. \cite{Sugi2016}).  One of the difficulties is that the conjecture describes very precisely the position of all the eigenvalues (or equivalently, the roots of the $G$ function) for a given parity and not only the large energy behavior. We remark that no counterexample to the conjecture has surfaced in our numerical experiments. We note that in the conjecture the eigenvalues are considered normalized, that is, translated by $g^2$ (i.e. $\lambda + g^2$).

\begin{conj} %\label{conj:G}
  For a given parity, the number of eigenvalues in an interval $[n,n+1)$ is either $0$, $1$, or $2$.
  Moreover, an interval with two eigenvalues cannot be adjacent to another interval with two eigenvalues.
  Similarly, an interval with no eigenvalues cannot be adjacent to another interval with no eigenvalues.
\end{conj}

Using the results of this paper and the numerical observation, we can make some remarks on the conjecture.
First, if $\alpha_0$ is bounded above by $1$ as we expect, then for full spectrum there should be
no intervals of the form $[n,n+1)$ that contain no eigenvalues.

Note that according to the $G$-conjecture there may be intervals with no eigenvalues  of a given parity. For the full spectrum, whenever such empty interval appears for one parity, an eigenvalue of the opposite parity must be present as in \eqref{eq:mixedcrossing}.

In fact, according to the Weyl law for the QRM (obtained using spectral zeta methods \cite{Sugi2016}) the number of
eigenvalues for the QRM smaller than $T>0$ should be $2T$ as $T \to \infty$, we expect that actually, all intervals
$[n,n+1]$ contain $2$ eigenvalues.  
In  \cite{R2023}, this result was proved in density one (that is, except for a finite number of intervals), making it the current best result for the $G$-conjecture to date. In particular, this also shows that the number of intervals containing a single eigenvalue is finite, or equivalently, that, asymptotically, intervals with two eigenvalues always alternate between the two parities.

%Note that due to presence of oscillatory terms in the approximation of Theorem \ref{thm:asymptotic} it is not easy to establish this kind of result (see Remark \ref{rem:exceptions}). We also remark here that the Weyl law is the same for  the AQRM \cite{R2020}.

Other questions about the distribution of the eigenvalues remain open.  The relation between the system parameters and the number of intervals with a single eigenvalue may also be an interesting topic of study. In addition, the extension of the $G$-conjecture for the AQRM.

\begin{comment}
\begin{rem}
\label{rem:exceptions}
    Let $0< \delta < \frac12$. In general, in the asymptotic approximation of Theorem \ref{thm:asymptotic} the order of the second term in \eqref{eq:asymptotics} dominates the error term. However, due to the presence of oscillation with the \emph{cosine} contribution which could be very small, we cannot discount the possibility of 
    having
    \[
        n -g^2 < \lambda_n(H_{\pm}), 
    \]
    and 
    \[
      \lambda_{n+1}(H_{+}) < n+1 - g^2 \qquad \text{ or } \qquad \lambda_{n+1}(H_{-}) < n+1 - g^2,
    \]
    in other words, we may have intervals $[n,n+1)$ with $3$ or $4$ eigenvalues. An effective version of Theorem \ref{thm:asymptotic}
    would allows us to give lower bounds to the proportion of intervals $[n,n+1)$ with exactly two eigenvalues. In Section
    \ref{sec:furtherobservations} using numerical estimates we consider the order of the error term in \eqref{eq:asymptotics}.
\end{rem}
\end{comment}
%%%%%%%%%%%%%%%%%%%%%%%%%%%%%%%%%%%%%%%%%%%%%
% Section 5: Further numerical observations %
%%%%%%%%%%%%%%%%%%%%%%%%%%%%%%%%%%%%%%%%%%%%%

\section{Further numerical observations}
\label{sec:furtherobservations}

In this section, we fix the bias parameter $\e$ and consider aspects of the energy spacing of the AQRM different from the asymptotic analysis of Section \ref{sec:spacing}.
Concretely, we introduce an internal symmetry observed numerically in the eigenvalue spacing, and we further describe the behavior of the first $[2\e]$ (low-energy) states and the corresponding spectral curves. This is closely related to the remarkably good approximations of low-energy states observed in \cite{RW2022}, using a \emph{single} particular polynomial that appears in the description of hidden symmetries and the spectral degeneracy of the AQRM for $\e\in \frac12\Z$, which appears to be outside the scope of perturbation theory. Finally, in Section \ref{sec:curvefitting} using curve fitting methods we make observations on the order of the error term on the approximation of Theorem \ref{thm:asymptotic}.

\subsection{Internal symmetry in the energy spacing}
\label{sec:internal_symmetry}

The asymptotic symmetry described in Section \ref{sec:refinement} is based on the variation of the bias parameter $\e$. In other words, it may be considered a symmetry between different models (corresponding to different system parameters). In addition, the numerical analysis also shows the existence of a symmetry in the energy spacing that occurs for any fixed set of system parameters, which we may call an {\em internal symmetry}.

\begin{figure}[!ht]
\centering
\includegraphics[width=0.75\textwidth]{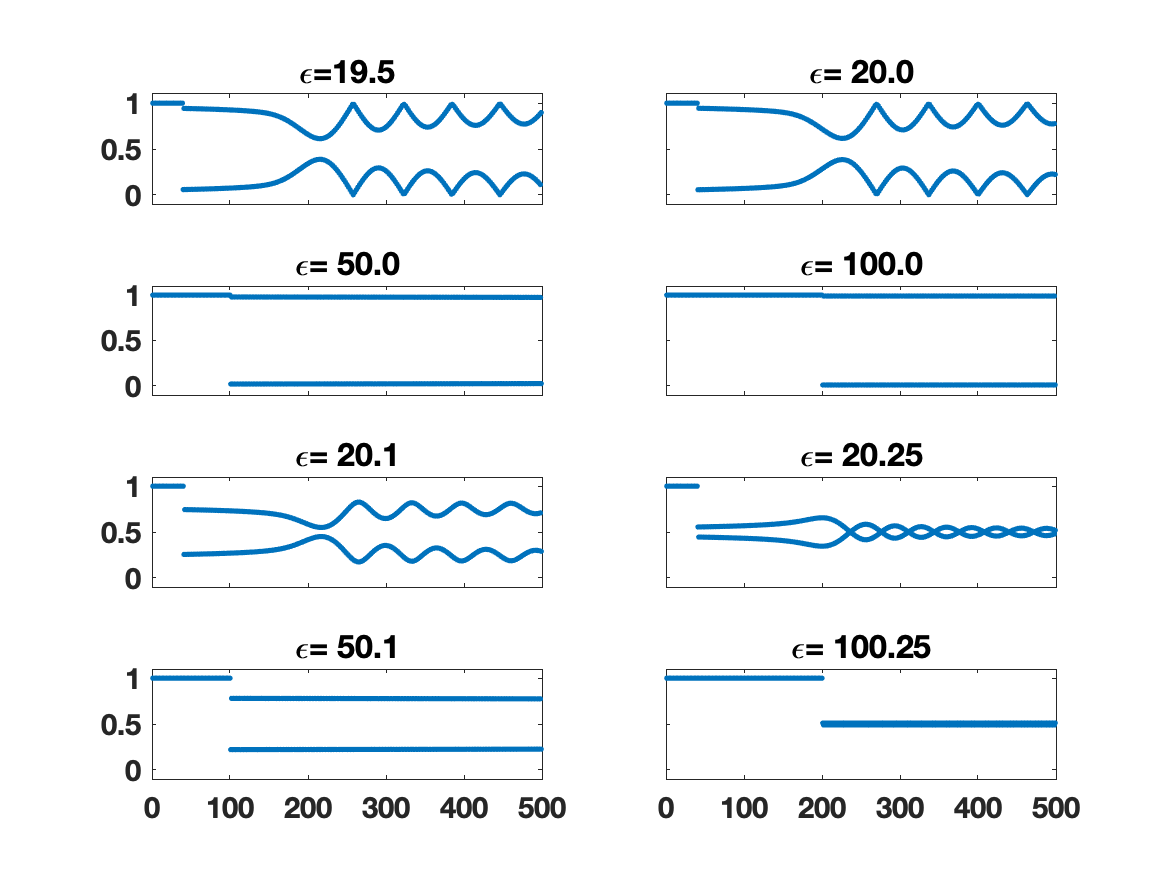}
\caption{Energy spacing of AQRM with $g=\Delta=1$ when $\epsilon$ is large. The apparent line segment of length $2[\e]$ in the upper left corner of each graph is related to the first $\ell$ spectral curves ($\e=\frac\ell2$, see the spectral graph shown in Figure \ref{fig:Eigencurves2} below). Note that the figures consist of many nearby points, not continuous lines.}
\label{fig:ES_AQRMs_large_epsilon}
\end{figure}

To illustrate this symmetry, in Figure \ref{fig:ES_AQRMs_large_epsilon} we show the graph of the first $500$ energy spacing points for different bias parameters $\e$.  First, we note that with the possible exception of a number of eigenvalues at the beginning (see Section \ref{sec:excellent approx} below), the energy spacing is contained in the intervals given by
\[
  \begin{cases}
    (\{\e\}, 1-\{\e\})  & \text{ if } 0 \leq \{\e\} < \frac14 \\
    (\frac12 - \{\e\}, \frac12 + \{\e\})  & \text{ if } \frac14 \leq \{\e\} < \frac12 \\
    (\{\e\}-\frac12, \frac32 - \{\e\})  & \text{ if } \frac12 \leq \{\e\} < \frac34 \\
    (1-\{\e\}, \{\e\})  & \text{ if } \frac34 \leq \{\e\} < 1
  \end{cases}.
\]
%where we recall that $\{x\}$ denotes the fractional part of $x \in \R$. 
Interestingly, when comparing the shape of graphs in Figure \ref{fig:ES_AQRMs_large_epsilon}, we see that the energy spacing appears to be symmetric with respect to $\frac12$ for each $\e$. 
Concretely, for sufficiently large $n$ the numerical energy points satisfy
\[
    \left( \lambda_n - \lambda_{n-1}\right) +   \left( \lambda_{n+1} - \lambda_{n}\right) \approx 1.
\] 
In Figure~\ref{fig:internal_symmetry} we illustrate the internal symmetry by reflecting the energy spacing points smaller than $\frac12$.

\begin{figure}[!ht]
\centering
\includegraphics[width=.8\textwidth]{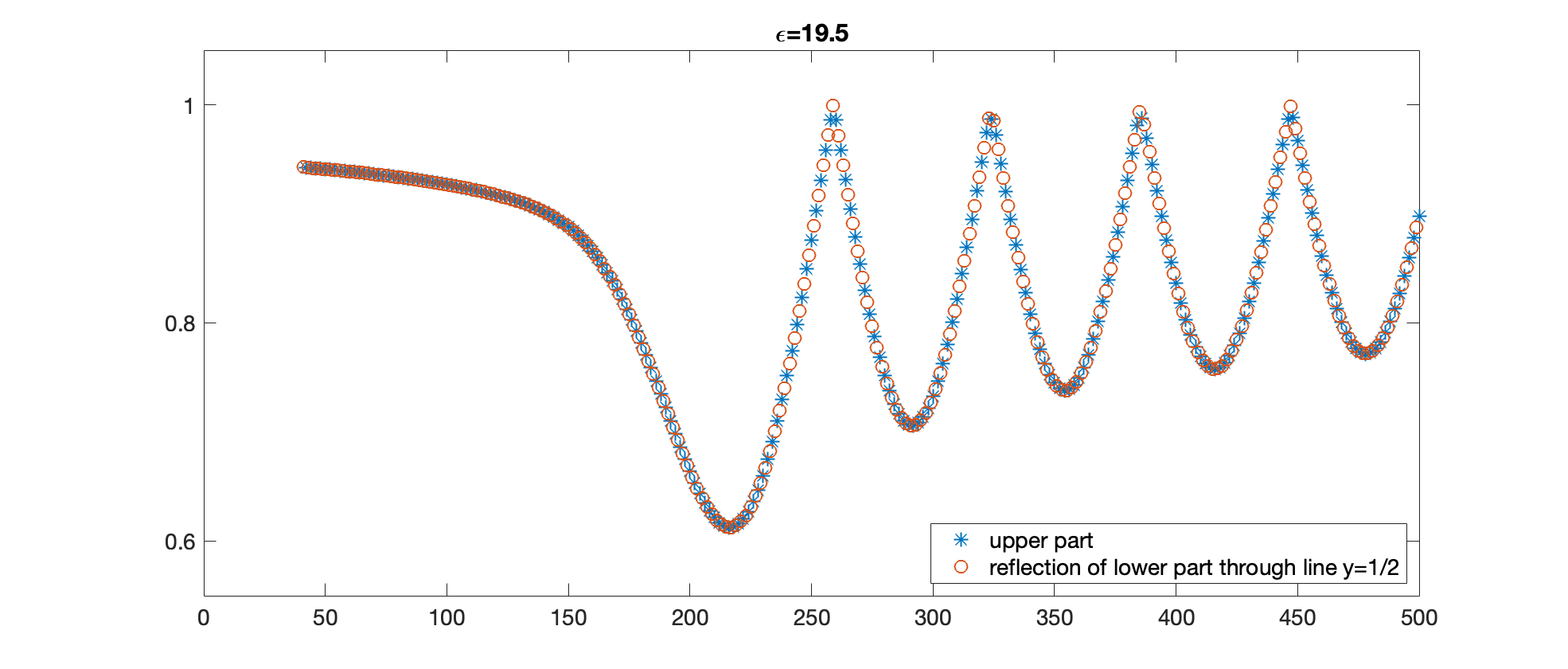}
\caption{Plot of spacing points where the spacing points smaller than $1/2$ are reflected.}
\label{fig:internal_symmetry}
\end{figure}

We note that the internal symmetry is consistent with the computational graphs in Figures \ref{fig:Distinguish_ES_on_Parities} and \ref{fig:truncated_density_vs_N}.  However, despite appearing in all of the numerical cases considered, this internal symmetry in the energy
spacing is difficult to explain in a theoretical way.

\begin{figure}[!ht]
\centering
\includegraphics[width=1\textwidth]{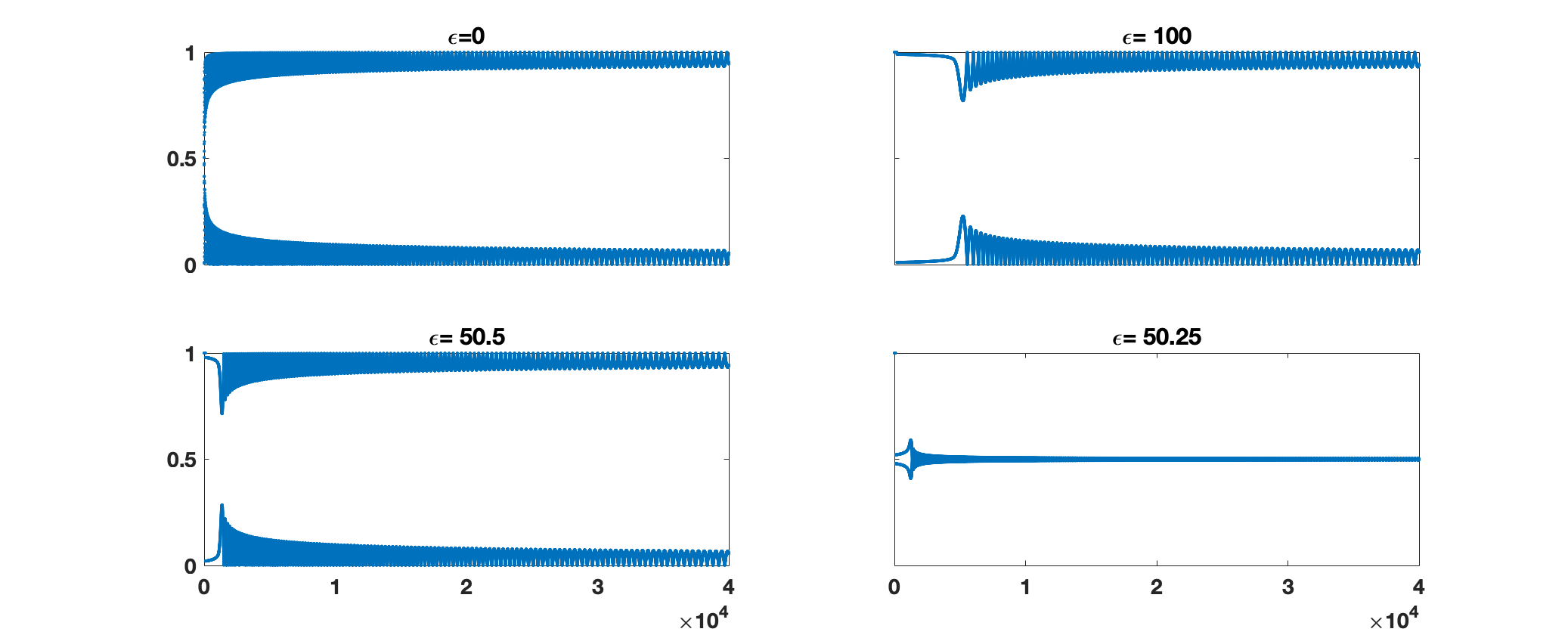}\\
\caption{Asymptotic behavior of energy spacing of QRM and AQRM with half-integer bias, $g=\Delta=1$.}
\label{fig:Internal_symmetry_1}
\end{figure}

\subsection{Analytic approximation of low-energy levels}
\label{sec:excellent approx}

In this subsection, we make additional observations from the degeneracy and hidden symmetry point of view of the behavior of the level spacing for large values of the bias parameter $\epsilon$ (shown in Figure~\ref{fig:ES_AQRMs_large_epsilon}), in particular regarding the energy spacing points that
do not appear to have the symmetric property discussed in Section \ref{sec:internal_symmetry}.

First, in Figures \ref{fig:ES_AQRMs_large_epsilon} and \ref{fig:Internal_symmetry_1}, we note that for $\epsilon=\frac{\ell}2 \in \frac12\Z_{\geq0}$, the first $\ell$ level spacing have values close to $1$. Numerically, these energy levels appear to be non-degenerate but a rigorous proof has not been given with the exception of the ground state, and, asymptotically for $g\rightarrow \infty$.

For half-integer $\e = \frac{\ell}2$, the spacing of the first $\ell$ energy levels follows from the conjectured excellent approximation of the lower energy levels by the zero locus of the polynomial $p_\ell(x;g,\Delta)$. Concretely, this is the conjecture that the points in the $(x, g)$-plane satisfying
\begin{equation}
  \label{eq:excellent}
  p_\ell(x;g,\Delta)=0
\end{equation}
correspond to the first $\ell$ energy eigenvalues of $\AHRabi{\ell/2}$ as shown in  Figure \ref{fig:Eigencurves2}.

\begin{figure}[h!]
  \begin{center}
    \includegraphics[height=4cm]{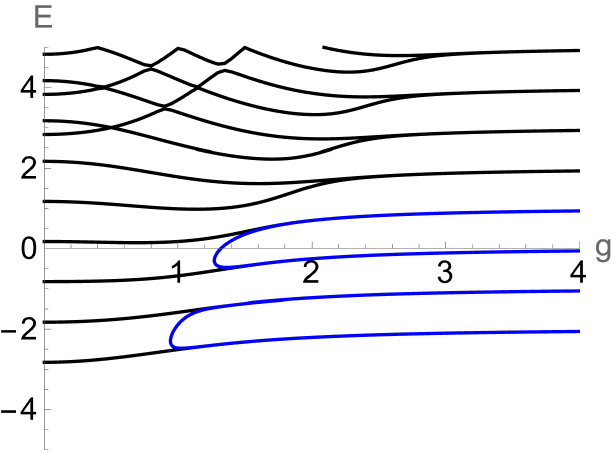}
    \quad
    \includegraphics[height=4cm]{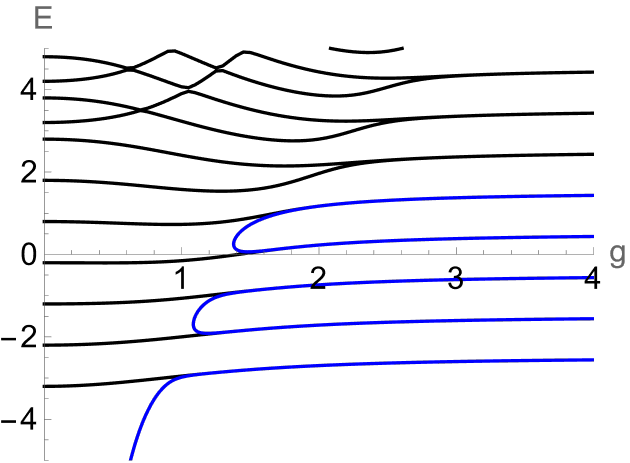}
  \end{center}
  \caption{Renormalized spectral curves (black) of the Hamiltonian $\AHRabi{\e}+g^2$ ($\e=\frac{\ell}2$), curves defined by $p_\ell(\text{E};g,\Delta)=0$ \eqref{eq:excellent} ($\text{E}$ being the eigenvalue of $\AHRabi{\e}+g^2$) and an excellent approximation by curves (blue) for $\Delta = 2$ for  $\ell=4$ (left) and $\ell=5$ (right). }
  \label{fig:Eigencurves2}
\end{figure}

This approximation is quite remarkable because it approximates, albeit implicitly, a whole set of eigenvalues in closed form for all sufficiently large coupling parameters. Physically, that the level spacing must be close to $1$ is a consequence of the fact that the low lying eigenstates live dominantly in the spin-down sector and the hybridization with the upper level via the coupling term is suppressed for large $\e$ -- a similar phenomenon is well-known in the QRM for large $\Delta$. The discontinuity in the level spacing seen in Figure~\ref{fig:ES_AQRMs_large_epsilon} is equivalent to a discontinuity in the density of states (DOS) and indicates an excited state quantum phase transition (see \cite{PHP2016} for the analogous case of the QRM with large $\Delta$).
In addition, in the weak limit $g\to \infty$ the spectrum is given by $\left\{ n \pm \frac{\ell}2  \,|\, n \geq 0\right\}$ (see \eqref{eq:limitspec}) which shows that the spacing for the first $\ell$ eigenvalues in this limit is exactly $1$. 

In fact, the numerical results above appear to be connected to the main conjecture in \cite{RW2022} which describes the connection between the degeneracy and hidden symmetry for the AQRM. Concretely, by the determinant expression of $A^\ell_N(u,v)$ given in \eqref{eq:Adetexp}, the main conjecture in \cite{RW2022} is equivalently formulated in the following form.

\begin{conj} %\label{conject:detexpP}
  The polynomial $p_{\ell}(x;g,\Delta)$  has an expression as 
  \begin{equation} 
    \label{eq:polyPandA}
  p_{\ell}(x;g,\Delta) = \det \left( \Delta^2 \bm{I}_\ell + \bm{M}_{\ell}(x,g) \right).
  \end{equation}
  Here $\bm{M}_{\ell}(x,g)$ is the tri-diagonal matrix given by 
  \[
    \bm{M}_{\ell}(x,g) = \Tridiag{((2g)^2-\ell+2i-1)(x - \tfrac{\ell}2 + g^2 + i) }{(x - \tfrac{\ell}2 + g^2 + i)}{-i(\ell-i)(x - \tfrac{\ell}2 + g^2 + i+1)}{1\le i\le \ell}.
  \]
\end{conj}

The excellent approximation given in Figure \ref{fig:Eigencurves2} is obtained by using the determinant expression in the right-hand side of \eqref{eq:polyPandA}. Since the conjecture gives the explicit formula of $p_{\ell}(x;g,\Delta)$, we see that the graph  $p_{\ell}(x;g,\Delta)=0$  in $(g, x)$-plane actually gives an excellent approximation for the curves exhibited in Figure \ref{fig:Eigencurves2}  together with the curve given by \eqref{eq:excellent}. We refer the reader to \cite{RW2022,RW2023} for more examples and further discussion. 

Moreover, since the existence of the line segment close to $1$ in the upper left of each graph for a half-integer case in Figures \ref{fig:ES_AQRMs_large_epsilon} and \ref{fig:Internal_symmetry_1}) are consistent with the excellent approximation for the first $\ell$ curves in Figure \ref{fig:Eigencurves2}, the two numerical results observed from Figures  \ref{fig:ES_AQRMs_large_epsilon} and \ref{fig:Eigencurves2} are considered more trustworthy for representing the situation.

In the general $\e$ case, we also observe a number of spacing points at the last four graphs in Figure \ref{fig:ES_AQRMs_large_epsilon} having spacing close to $1$ and the number appears to be $2[\e]$, where $[x]$ denotes the integer part of $x$. 

The persistence of this behavior suggests that there might be a way to interpolate the curves \eqref{eq:excellent} between different values of $\ell$ while keeping the excellent approximation shown in Figure \ref{fig:Eigencurves2}. Precisely, we may infer that graphs of the first $m=2[\e]$ spectral curves look similar that we observe in Figure \ref{fig:Eigencurves2} (loosely speaking, parallel curves with an interval width of $1$ approximately), while there are no level crossings (not like those graphs in Figure \ref{fig:Eigencurves2}) for $n$-th spectral curves for $n>m$. However, even if such curves (given by equations similar to  \eqref{eq:excellent}) exist for a non-integral $\e$, and since there is no symmetry or degeneracy of the AQRM in these cases, it must be given by an analytic (and transcendental) function, that is, by a power series expression (not a polynomial).

We remark that for small values of parameters $g$ and $\Delta$ the behavior of the spacing of the first energy points may be explained by using Rayleigh-Schr\"odinger perturbation theory (see e.g. \cite{SK2017}), but the excellent approximation property is stronger since it appears to hold for arbitrary values of the system parameters.

\begin{rem}
We recall that the constraint polynomials have also been used in the generalized adiabatic approximation \cite{LB2021b} to improve the approximation of the spectral curves (not limited to the low energy spectral curves ) given by the usual adiabatic approximation that relies on Laguerre polynomials. All things considered, it appears that the constraint polynomials for the exceptional spectrum have a deep relation to the full spectrum, including the regular part, which merits further studies.
\end{rem}

\subsection{Energy spacing distribution for fixed parity}
\label{sec:curvefitting}

In this subsection, we investigate the large-energy behavior of energy spacing for a single parity of the QRM via the curve fitting method to complement Theorem \ref{thm:asymptotic} of \cite{BZ2021} and give further insight into the speed of convergence.

We consider the QRM with $g=\Delta=1$ and positive parity. The statistics of the energy spacing is determined by the high-energy part where the spacing becomes regular. In this calculation, we consider spacing from $\lambda_{10,000}$ to $\lambda_{45,550}$ and pay attention only to values that are greater than or equal to $1$. When we connect the peaks (red points in Figure~\ref{high-energy-part}(a), we observe a ``smooth'' curve as in Figure~\ref{high-energy-part}(b). Therefore, it is most likely that we can describe this curve, the envelope of the spacing distribution, by some approximate analytical expression. We also calculate the period lengths by taking the distances between consecutive peaks. Apparently the period increases for increasing $\lambda_n$. More precisely, the width of the periods increases almost linearly with the index of the period, as seen in Figure~\ref{high-energy-part}(c)).

\begin{figure}[ht]
\centering
\includegraphics[width=11cm]{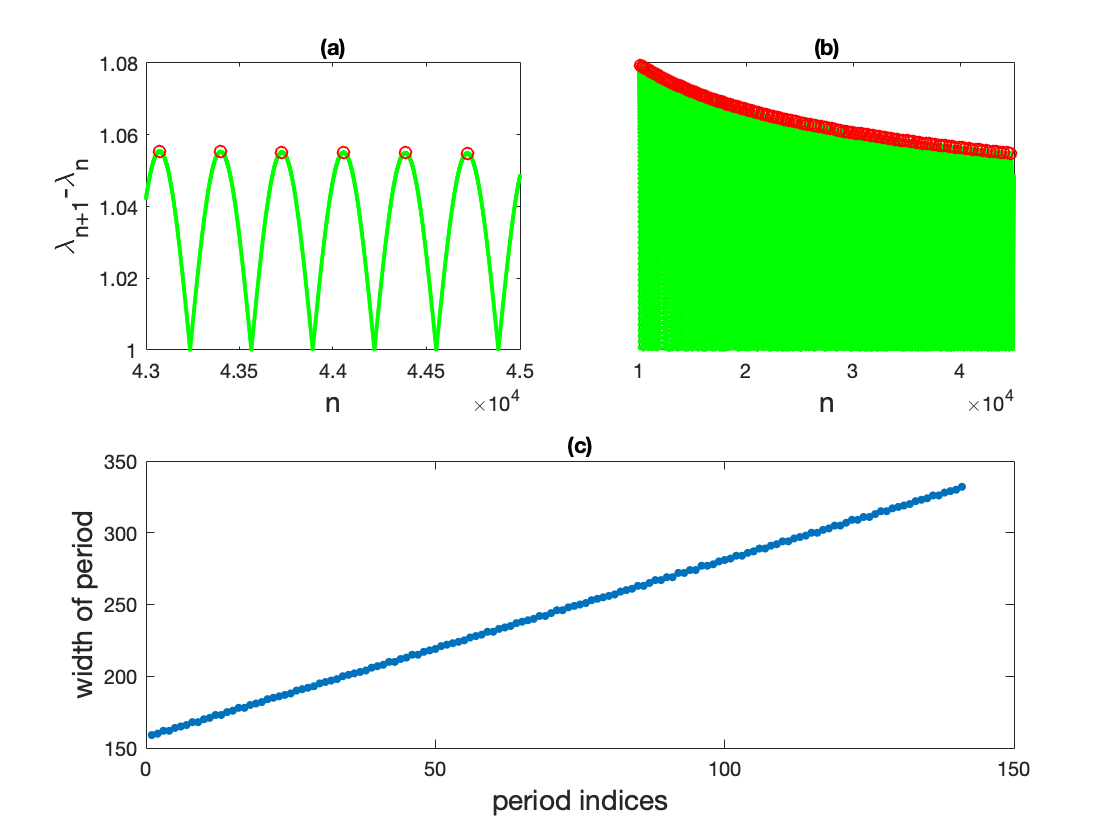}
\caption{Envelope of upper part (values greater than or equal to $1$) on energy spacing graph (a) and (b). The period increases almost linearly with the period index (c).}\label{high-energy-part}
\end{figure}

Observe the high energy parts on each parity, we see that they behave very similar to a Bessel function of the first kind $J_{\alpha}(x)$  (see e.g. \cite{AAR1999}), or the simple function
\begin{equation}
\label{eq:Harmonic_Oscillator}
x(t) = A_0 \exp^{-\epsilon \omega_0 t}\sin(\sqrt{1-\epsilon}\omega_0t + \phi).
\end{equation}

However, in the Bessel function of the first kind and \eqref{eq:Harmonic_Oscillator}, the period remains constant, while the distribution of the level spacing has increasing period length, see Figure~\ref{high-energy-part}(c). 
Because of this complex behavior of the energy spacing on the positive/negative parities, it is not possible to approximate the behavior of the energy spacing on parities using the Bessel function or \eqref{eq:Harmonic_Oscillator}. We need another approach. 

We have a sequence of peaks (the red points on Figure~\ref{high-energy-part}(a,b) that we assigned to the variable $y$. We construct a variable $x$ containing the enumerated indices of the elements in $y$ starting from 0 to the length of the sequence $y$.  After applying the translation of axes 
\[
    (x,y)\to (X=x,Y=y-1), 
\] 
utilizing the curve fitting method, we discovered that the envelope may be modelled by a function of the form 
\begin{equation*}
%\label{eq:power_fitting}
    p(n)=a n^b,
\end{equation*}
with coefficients estimated with 95\% confidence bounds as
\[
  a=0.798, \quad (0.7929, 0.7932),  \quad \text{ and } \quad b=-0.2495, \quad (-0.2495,-0.2495). 
\]
The goodness of fit scores are given by
\[
    \text{\rm {SSE}}=9.6798\times 10^{-10}, \text{ and } R\text{-square} = 1.0000.
\]
Clearly, since $b<0$, we have
\[
  \lim_{n\to\infty}p(n)=\lim_{n\to\infty}a n^b=0\qquad  %{\rm in} %\,\,(X,Y){\rm - coordinate},
\]
which is consistent with \eqref{eq:spacing single parity} and, moreover, is evidence that the exponent $-\tfrac14$ in the oscillating term in Theorem \ref{thm:asymptotic} may not be further reduced and that true order of magnitude of the error term may be smaller.

Let us now give an approximation for the period of the oscillation for $g=\Delta=1$. First, according to Theorem \ref{thm:asymptotic}, the distance between two consecutive levels with positive parity reads
\[
    \lambda_{n+1}-\lambda_n=1+(-1)^n\frac{1}{\sqrt{2\pi}}\left(
    \frac{\cos(4\sqrt{n+1}-\pi/4)}{(n+1)^{\frac14}}+ \frac{\cos(4\sqrt{n}-\pi/4)}{n^{\frac14}} \right)
    +O(n^{-1/2+\delta}).
\]
For large $n$, we write $n=n_0+2\tilde{n}$ with $n_0\gg \tilde{n} > \sqrt{n_0} \gg1$.
By neglecting the contribution of the error term of order $n^{-1/2+\delta}$, we approximate the oscillatory term above 
by
\begin{equation}
    \label{deltal}
      \frac{(-1)^{n_0}(n_0+2\tilde{n})^{-1/4}}{\sqrt{2\pi}}
      \left(\cos(4\sqrt{n_0+1+2\tilde{n}}-\pi/4)+\cos(4\sqrt{n_0+2\tilde{n}}-\pi/4)\right),
\end{equation}
obtained by noting that $(n+1)^{-1/4} \approx n^{-1/4}$.
Then, the expansion of the square roots in the cosines in \eqref{deltal} yields the approximation
\begin{align}
    &\cos\left(4\sqrt{n_0}\left[1+\frac{1}{2}\frac{2\tilde{n}+1}{n_0}\right]-\pi/4\right)
    +\cos\left(4\sqrt{n_0}\left[1+\frac{1}{2}\frac{2\tilde{n}}{n_0}\right]-\pi/4\right)\nonumber\\
    &=\cos\left(C_1(n_0) + \frac{4\tilde{n}}{\sqrt{n_0}}\right) +\cos\left(C_2(n_0) + \frac{4\tilde{n}}{\sqrt{n_0}}\right).
    \label{appr}
\end{align}
In the variable $2\tilde{n}$, the period of this formula is given by
\[
  L(n_0)=\pi\sqrt{n_0},
\]
showing that period length increases proportionally to the square root of $n$, as observed
numerically in this section.
Moreover, if the period index is written as the continuous variable $\xi$, we have
\[
    n-n_{min} = \pi\int_{\xi(n_{min})}^{\xi(n)}\sqrt{n(\xi)}\textrm{d}\xi,
\]
where $n_{min}$ is the minimal value of $n_0$ where the approximation in \eqref{appr} becomes valid. Immediately, this
entails
\begin{equation}
    \frac{ {\rm d} n}{ {\rm d} \xi}=\pi\sqrt{n(\xi)}, \qquad \xi(n) = \frac{2}{\pi}\sqrt{n}+C(n_{min})=\frac{2}{\pi^2}L(n)+C(n_{min}),
    \label{rel1}
\end{equation}
that is, there can be seen an almost linear relation between period and period index, as depicted in Figure~\ref{high-energy-part}(c). 
In fact, as we did not consider the alternating factor in \eqref{deltal} for this analysis the actual period 
$\bar{L}(n)$ and index $\bar{\xi}(n)$ shown in Figure~\ref{high-energy-part}(a) are given by
\[
  \bar{L}(n) = \frac12 L(n), \qquad \bar{\xi}(n) = 2 \xi(n),
\]
and thus, from \eqref{rel1} we conclude that 
\begin{align}
  \label{eq:period_approx}
  \bar{\xi}(n)=\frac{8}{\pi^2} \bar{L}(n)+C.
\end{align}

We now compare with the numerical data where we have a total of 143 periods. Here, assuming a linear relation between the (observed) period index $\bar{\xi}(n)$ and the period $\bar{L}(n)$, we obtain using curve fitting the equation
\begin{equation*}
%\label{period_estimation}
\bar{L}(n) = 1.234 \cdot \bar{\xi}(n) + 157,
\end{equation*}
with $\text{\rm SSE}=21.9059$ and $R\text{-square}=0.9999$. Furthermore, with $95\%$ confidence, the coefficient of period index belong to the interval $(1.231, 1.235)$ and the range of the intercept is $(156.9, 157.1)$.

Since $\pi^2/8 \approx 1.2337$ this gives almost a perfect agreement of the approximation \eqref{eq:period_approx} with the data, further supporting the notion that the order of the oscillatory term in Theorem \ref{thm:asymptotic} is the best possible and that order of the error term may be improved.

%%%%%%%%%%%%%%%%%%%% 
% Acknowledgements %
%%%%%%%%%%%%%%%%%%%%

\section*{Acknowledgements}

The authors would like to thank Zeev Rudnick for useful comments on an earlier version of the manuscript.

This work was partially supported by JSPS Grant-in-Aid for Scientific Research (C) No.20K03560, JST CREST JPMJCR2113, Japan. LTHN acknowledges the support by Daicel Corp., Japan. DB is funded by the German Research Foundation (DFG) under grant No. 439943572.

% if needed...
% All the authors contributed equally to this work.   

%%%%%%%%%%%%%%%%
% Bibliography %
%%%%%%%%%%%%%%%%

\smallskip

\flushleft

Daniel Braak \par
\textsc{Department of Physics, Augusburg University,\\ Universit\"asstr. 1, 86159, Augsburg, Germany}

\texttt{daniel.braak@physik.uni-augsburg.de}

\bigskip

Linh Thi Hoai Nguyen \par
\textsc{International Institute for Carbon-neutral Energy Research, Kyushu University, \\ 744 Motooka Nishi-ku, Fukuoka 819-0395, Japan}

\texttt{linh@i2cner.kyushu-u.ac.jp}

\bigskip

Cid Reyes-Bustos \par
\textsc{NTT Institute for Fundamental Mathematics, \\NTT Communication Science Laboratories, Nippon Telegraph and Telephone Corporation, \\3-9-11 Midori-cho Musashino-shi, Tokyo, 180-8585, Japan}

\texttt{cid.reyes@ntt.com}

\bigskip

Masato Wakayama \par
\textsc{NTT Institute for Fundamental Mathematics,\\ NTT Communication Science Laboratories, Nippon Telegraph and Telephone Corporation,\\ 3-9-11 Midori-cho Musashino-shi, Tokyo, 180-8585, Japan}

\texttt{wakayama@imi.kyushu-u.ac.jp, masato.wakayama@ntt.com}

\end{document}